\documentclass[review, authoryear]{article}
\usepackage{graphicx}
\usepackage{bm}
\usepackage{subfigure}
\usepackage{float}
\usepackage{rotating}
\usepackage{epstopdf}
\usepackage[]{caption2}
\usepackage{latexsym,ifthen,amssymb}
\usepackage{natbib}
\usepackage{algorithm}
\usepackage{algorithmicx}
\usepackage{algpseudocode}
\usepackage[toc,page,title,titletoc,header]{appendix}
\usepackage{multirow}
 \usepackage{longtable}
 \usepackage{rotating}
 \usepackage{graphicx, amsmath, amsthm, amssymb,float,setspace,color, multirow}
\usepackage{enumitem}

\newtheorem{theorem}{Theorem}[section]
\newtheorem{proposition}[theorem]{Proposition}
\newtheorem{lemma}[theorem]{Lemma}
\newtheorem{corollary}[theorem]{Corollary}
\newtheorem{assumption}[theorem]{Assumption}

\newtheorem{question}[theorem]{Question}

\theoremstyle{definition}

\theoremstyle{remark}

\title{
Uniform Consistency in
Stochastic Block Model with Continuous
Community Label}

\author{Lu Liu, Lili Wang, Qingsong Xu}

\author{\large
\textsc{Lu Liu\footnote{ g.jiayi.liu@gmail.com}},
Lili Wang, Qingsong Xu
\\ \textit{Central South University}
 }

\def\mbbP{\mathbb{P}}
\def\mbfP{\mathbf{P}}
\def\h{\hat}
\def\oE{\operatorname{E}}
\def\al{\eta}
\def\mb{\mathbf}
\def \mtB{\mathcal{\tilde{B}}}
\def\b{\boldsymbol}
\def\t{\tilde}
\def\bhtal{\b{\hat{\t{\al}}}}
\def\btal{\b{\t{\al}}}
\def\1u{\mathbf{1}_{\mb{u}}}

\def\klowercase2{
8\cdot \frac{n\mu_n^*\delta_n^2}{4\sqrt{d}C_\varepsilon C_\mathcal{B}}}

\def\kupper{
\frac{n\delta_n^2}
 {2e^3 Cs\sqrt{d}C_\varepsilon C_{\mathcal{B}}}
}

\def\conditionofdelta{
\frac{n\mu_n^*\delta_n^2}{8\sqrt{d}C_\varepsilon C_\mathcal{B}}
\geq 20\log n}

\def\conditionofdeltasecond{
\delta_n
\geq \sqrt{
\frac{8\sqrt{d}C_\varepsilon C_\mathcal{B}
20\log n}
{n\mu_n^*}
}}

\def\conditionofdeltathird{
\frac{n\delta_n^3}
 {2e^3 Cs \sqrt{d}C_\varepsilon C_{\mathcal{B}}}
}

\def\conditiononmuth6{
Cs\mu_n^*<\frac{1}{8}}

\def\const1lem6{
\frac{n\mu_n^*\delta_n^2}{2\sqrt{d}C_\varepsilon C_\mathcal{B}}
}

\begin{document}
\maketitle
\begin{abstract}
\cite{bickel2009nonparametric} developed a general framework to establish consistency of community detection in stochastic block model (SBM). In most applications of this framework,  the community label is discrete. For example, in \citep{bickel2009nonparametric,zhao2012consistency} the degree corrected SBM is assumed to have a discrete degree parameter. In this paper, we generalize the method of \cite{bickel2009nonparametric}  to give consistency analysis  of maximum likelihood estimator (MLE) in SBM with continuous community label. We show that there is a standard procedure to transform the $||\cdot||_2$ error bound to the uniform error bound $||\cdot||_\infty$. We demonstrate the application of our general results by proving the uniform consistency (strong consistency) of the MLE in the exponential   network model with   interaction effect.   As far as we know, this is  yet unknown. Unfortunately, in the continuous  parameter case,  the condition  ensuring uniform consistency  we obtained is much stronger  than that in the discrete parameter case,  namely  $n\mu_n^5/(\log n)^{8}\rightarrow\infty$ versus  $n\mu_n/\log n\rightarrow\infty$. Where  $n\mu_n$ represents the average degree of  the network.  But continuous is the limit of discrete. So it is not surprising  as we show  that by discretizing the community label space into sufficiently small (but not  too small) pieces and applying the MLE on the  discretized community label space,  uniform consistency holds under almost the same  condition as in discrete community label space.  Such a phenomenon is surprising since the discretization does not depend on the data or the model. This reminds us of the thresholding method. Another purpose of this paper  is to investigate whether the uniform   consistency condition in the continuous community label case is necessarily stronger than  that in the discrete parameter case. We did not find a numerical example  confirming this. However, a numerical experiment shows that the uniform error bound and the mean square error bound of the MLE, approximated by  the gradient decent algorithm, are not reduced by strengthening the stopping criterion.  This coincides with the result that  running the MLE on a subset of the parameter space reduces the uniform error bound.

  \end{abstract}

\tableofcontents

\section{Introduction}

SBM is one of the most popular network model,
 not only because of its simplicity,
 but also for the following reasons.
First, it well fits a lot of 
real 
world data in the following fields,
 social network
 \citep{holland1981exponential,Newman2002Random,Robins2009Recent}
 (notably, \cite{holland1981exponential}
 first proposed SBM),
 biology \cite{rohe2011spectral},
gene regulatory network
\citep{Schlitt2007Current,Pritchard2000Inference},
image processing
\citep{Shi2000Normalized,Sonka2008Image}.
Second,
the model
is a nice tool to investigate
community detection algorithms from
the theoretical
perspective.
 \citep{Dyer1989The,Jerrum1998The,Condon2001Algorithms}
 are early works in this stream. Although
 their
 focus is the algorithmic aspects of
  the min-bisection problem.
Later a vast amount of research is carried out
to study and compare the performance of various
community detection
algorithms on SBM. Roughly speaking,
these algorithms
can be divided into the following categories.
Modularity algorithm
\citep{Newman2004Finding} etc,
likelihood algorithm
\citep{bickel2009nonparametric,Choi2012Stochastic,amini2013pseudo,Celisse2012Consistency} etc,
and most importantly,
spectral algorithm
\citep{chatterjee2011random,balakrishnan2011noise,Jin2014Fast,Sarkar2013Role,Krzakala2013Spectral} etc.

Notably, \cite{bickel2009nonparametric} provided
a general framework to check the
consistency of community detection.
It was further extended by \citep{zhao2012consistency}
to establish consistency of many
community detection algorithms
in more general models.
 These algorithms include
maximum likelihood estimation and
various modularity methods.
The technique is mainly based on finite covering method
and concentration inequality.
This approach is also employed to establish consistency
of spectral clustering \citep{Lei2015Consistency}.
Most of works concerning consistency
of parameter estimation in SBM
focused on discrete community label
(finitely many blocks).
A few exceptions include: \citep{yan2015asymptotics}
dealing with exponential network model;
\citep{xu2014edge},\citep{jog2015information}
 dealing with SBM with edge label
 (edge weight); and
\citep{continuousmembership2016}
 studying mixed membership. 
The general picture is still lack. Whether
uniform consistency and weak consistency
holds in SBM with continuous membership
under similar condition as that in discrete
membership case?
On the surface, continuous community label
space is simply the limit of discrete
community label space. However,
the  uniform consistency
(strong consistency) condition we obtained
in continuous community label space
is stronger than that in
discrete community label space. It is not
clear if such gap is inevitable.

\subsection{Outline}

The paper is organized as follows.
We define the SBM in section \ref{asymsubsec1}
and  introduce notations in section \ref{asymsubsec3}.
In section \ref{asymsecmainresult}
we  define uniform consistency and weak
consistency.
In section \ref{asymsubsec4}
we first give the
condition of weak consistency
and uniform consistency
for the exponential network model.
As in \citep{bickel2009nonparametric,zhao2012consistency},
the condition is about the density
(average connection probability)
of the network.
Then in section \ref{asymsubsec4}
we present results for weak consistency
and uniform consistency in the
general SBM. The condition for
weak consistency is slightly stronger than
 that in the discrete community label space.
However, the
uniform consistency condition we obtained
is much stronger than that in
the discrete community label space in
\citep{bickel2009nonparametric,zhao2012consistency}.
We are not clear if the condition is
necessarily stronger.
However, continuous is the limit
of discrete. Therefore, it is not
surprise as we show
in section \ref{asymsubsec5}
that by discretizing the
community label space and applying the MLE on the
discretized community label space,
the condition ensuring uniform consistency
is almost the same as
that in discrete community label space in
\citep{bickel2009nonparametric,zhao2012consistency}.
 In section \ref{asymsecsimulation},
we present our simulation results.
They show that
our theoretical results are possibly far from the truth.
First, we
run gradient decent algorithm
for the MLE in the exponential network
model. The results indicate that
the error of the MLE is asymptotically normal
and uniformly consistent. %
We have not found an example verifying
that the condition ensuring
uniform consistency in continuous
community label space is necessarily stronger
than that in the discretized community label
space. But our second experiment
shows that the performance of
the gradient decent algorithm,
measured by the uniform error bound,
can be reduced by \emph{relaxing}
the stopping criterion. This is an evidence that
the MLE can be overfitting.

\subsection{Model description}
\label{asymsubsec1}

A network of size $n$
is determined by
$n(n-1)/2$ variables
$Y_{ij}\in \mathcal{S},i<j\leq n$.
$Y_{ij}$ represents the connection type
between node $i,j$.
For an undirected  network
$\mathcal{S} = \{0,1\}$;
for a directed network $\mathcal{S} =
\{0,1\}^2$. More complicated state space
is  possible \citep{lelarge2015reconstruction}. In all these
settings, there exists an element
$\mathbf{0}\in \mathcal{S}$ indicating
"no connection". In a stochastic
block model (SBM), each node $i$ is assigned
a latent type $\boldsymbol{\al}_i$.
Conditional on $\boldsymbol{\al}$,
$Y_{ij}$ are mutually independent,
and $\mbbP(Y_{ij}=y|\boldsymbol{\al})
 = \Theta_{ij}(y;\boldsymbol{\al}_i,\boldsymbol{\al}_j,\boldsymbol{\rho},
 \mu_n)$
 where $\boldsymbol{\rho}$ is the model parameter
 and $n\mu_n$ reflects the order
 of average degree.
 In a typical SBM consisting of $K$ communities,
 $\boldsymbol{\al}_i\in\{1,2,\cdots,
 K\}$ represents the community label of $i$
 and  $\boldsymbol{\rho}$ is a $K\times K$
  matrix characterizing the between
 community connection probability. In this paper,
 we study the case where $\boldsymbol{\al}_i$
 can be continuous variable.
 In many settings, the network is somewhat sparse,
 which requires $\Theta_{ij}$ to scale with the network
 size $n$. We set
 \begin{align}\label{asymmodel2}
 &\Theta_{ij}(y;\boldsymbol{\al}_i,\boldsymbol{\al}_j,\boldsymbol{\rho},\mu_n)
  = \mu_n f(y;\boldsymbol{\al}_i,\boldsymbol{\al}_j,\boldsymbol{\rho})
 \text{ if }y\ne\mathbf{0};\\ \nonumber
 &\Theta_{ij}(\mathbf{0};\boldsymbol{\al}_i,\boldsymbol{\al}_j,\boldsymbol{\rho},\mu_n) =
 1-\sum\limits_{y} \Theta_{ij}(y;\boldsymbol{\al}_i,\boldsymbol{\al}_j,\boldsymbol{\rho},\mu_n),
 \end{align}
 where $f$ is a fixed function.
 Model (\ref{asymmodel2}) clearly includes
 the traditional SBM (finitely many communities),
 degree-corrected SBM and a special
 case of the exponential network model
 (see section 2.1 or \citep{yan2015asymptotics}).
In many settings,
there is a model identification problem
concerning $\mu_n$.
We restrict on two settings
of $\mu_n$:
\begin{itemize}
\item $n\mu_n$ is
 the average "degree", i.e.,
$\mu_n=\mbbP(Y_{ij}\ne \mathbf{0})$;

\item $\mu_n$ is  known and
represents a level of density of the
network
say, $\mu_n = \frac{\log n}{n}$,
$\mu_n=1$ etc.

\end{itemize}
In both cases we write
$\overline{\mu}_n$ for the estimator of
$\mu_n$. There may also be a model identification
problem concerning $\b{\al},\b{\rho}$.
For example, in the SBM with finitely many
blocks, a permutation of community labels
gives the same model.
For two parameters $\b{\al},\b{\rho}$
and $\b{\al}',\b{\rho}'$, we write
$\b{\al},\b{\rho}\cong \b{\al}',\b{\rho}'$
if for all $i,j\leq n$,
$\Theta_{ij}(\cdot;\b{\al}_i,\b{\al}_j,\b{\rho},\mu_n)
= \Theta_{ij}(\cdot;\b{\al}_i',\b{\al}_j',\b{\rho}',\mu_n)$.
We say
a parameter $\b{\al},\b{\rho}$ is \emph{true}
iff $\b{\al},\b{\rho}\cong \b{\al}^*,\b{\rho}^*$.
Where $\b{\al}^*,\b{\rho}^*$ are realized true parameter.
In the following text, we denote by $\boldsymbol{\al}^*,
\boldsymbol{\rho}^*,\mu_n^*$  any true values
of $\boldsymbol{\al},\boldsymbol{\rho},\mu_n$.
Two major tasks in network data
analysis are
\emph{community detection}---
estimating $\b{\eta}$, and
\emph{parameter estimation}
---estimating $\b{\rho}$.
We call them both parameter estimation
since we regard $\b{\eta}$ also as
parameters.
 Clearly, the log likelihood function is,
$$
\ell(\mathbf{Y};\boldsymbol{\al},\boldsymbol{\rho},\mu_n)
= \sum\limits_{i<j\leq n}
 \log \Theta_{ij}
 (Y_{ij};\boldsymbol{\al}_i,\boldsymbol{\al}_j,\boldsymbol{\rho},\mu_n).
 $$
We define the \emph{maximum likelihood estimator},
namely $\h{\b{\al}},\h{\b{\rho}}$, to be
the maximizer of
$\ell(\mathbf{Y};\boldsymbol{\al},\boldsymbol{\rho},\overline{\mu}_n)$.

\subsection{Notations}
\label{asymsubsec3}
Let $\mathbf{Y}$ denote
the vector $(\cdots,Y_{ij},\cdots)$,
$i<j\leq n$,
$\mb{Y}_i$ denote
the vector $(Y_{i1},Y_{i2},\cdots,
Y_{in})$.
For a random variable $W$
with distribution $p_\theta$,
 write $\operatorname{E}_{\theta'}(W)$
 for
 $\operatorname{E}_{W\sim p_{\theta'}}(W)$;
 write $\oE_W(g(W))$
 ($Var_W(g(W))$) to denote the expectation
 (variation)
 with respect to $W$ (when no ambiguity
 is made the subscript is omitted).
For any
 vector $\mathbf{v}\in\mathbb{R}^n$,
let $\mathbf{v}_{-i}$ denote the
vector
$(v_1,\cdots,v_{i-1},
v_{i+1},\cdots,v_n)$;
  abusing the notation,
for any real $x$
write $
\mathbf{v}+x$ for
$(v_1+x,\cdots,v_n+x)$;
for
a sequence $X\in \{0,1\}^n$,
write $||\mb{v}||_{p,X}$
for
$
(\sum\limits_{i:X(i)=1}
|v_i|^p)^{\frac{1}{p}}$.

Let
\begin{align}
&\ell_i(\mathbf{Y}_i;\boldsymbol{\al},\boldsymbol{\rho},\mu_n)
 = \sum\limits_{j\leq n}
 \log \Theta_{ij}(Y_{ij};\boldsymbol{\al}_{i},\boldsymbol{\al}_j,\boldsymbol{\rho},\mu_n),
 \\ \nonumber
 &\ell_{ij}(Y_{ij};\boldsymbol{\al},\boldsymbol{\rho},\mu_n)
 =
 \log \Theta_{ij}(Y_{ij};\boldsymbol{\al}_{i},\boldsymbol{\al}_j,\boldsymbol{\rho},\mu_n).
\end{align}
We write $\hat{\boldsymbol{\rho}},\hat{\boldsymbol{\al}}$
 for the MLE  of the
parameters if not claimed otherwise.

For two sequences $a_n,b_n$,
we write $a_n = \Omega(b_n)$ iff
$b_n=O(a_n)$; $a_n=\Theta( b_n)$ iff
$a_n=O(b_n),b_n=O(a_n)$.
For two sequences of function
$a_n(x),b_n(x)$ we say
$a_n(x)=O(b_n(x))$ uniformly
in $x$ iff
there exists $C$ such that
$a_n(x)\leq C b_n(x)$ for all
$x$, and similarly for
$\Omega, O_p,o_p$ notation.

\section{Main Result}
\label{asymsecmainresult}
In the finite block setting,
there are two main concepts
 of consistency for community detection,
  weak consistency and uniform consistency
  (strong consistency).
 The weak consistency means that with large probability,
 the estimator of latent types is mostly correct.
 The uniform consistency means that
 with large probability,
 the estimator of latent types is all correct.
 In the continuous parameter setting,
 the appropriate measure for
 weak and uniform consistency are
 $||\cdot||_p$ norm and
 $||\cdot||_{\infty}$ norm
 where $0<p<\infty$. In particular,
 we adopt $||\cdot||_2$ as the
 weak consistency measure. i.e.,
 $\h{\b{\al}}$ is weak consistent
 iff $\inf\limits_{\b{\eta},\b{\rho} \text{ true}}
 ||\b{\h{\al}}-\b{\al} ||_2\xrightarrow{p} 0$;
  $\h{\b{\al}}$ is uniform consistent
  iff $\inf\limits_{\b{\eta},\b{\rho} \text{ true}}
 ||\b{\h{\al}}-\b{\al} ||_\infty\xrightarrow{p} 0$

\subsection{Consistency in exponential network model}
\label{asymsubsec4}
For a directed network model
write $a_{ij} = 1$ to indicate
the existence of an edge from
$i$ to $j$.
The exponential network model with interaction
effect we study
is a directed network model
with each node's latent type
being $\alpha_i,\beta_i\in\mathbb{R}$.
The distribution of $Y_{ij}=(a_{ij},a_{ji})$,
given $\alpha_i,\beta_i,i\leq n,\rho\in\mathbb{R}$, is
as follows:
\begin{align}\label{asymmodel1}
&g(a_{ij},a_{ji};
\alpha_i,\beta_i,\alpha_j,\beta_j,\rho)
 \\ \nonumber
 = &
\frac{1}
{
Z(\alpha_i,\beta_i,\alpha_j,\beta_j,\rho)
}
\exp\big\{\big[
(\alpha_i+\beta_j) a_{ij}+
(\beta_i+\alpha_j) a_{ji}\ \big]
+\rho a_{ij}a_{ji}
\big\}.
 \end{align}
 In this model we set $\mu_n^* = 1$.
There is clearly a model identification
problem. The models defined
by $(\boldsymbol{\alpha},
\boldsymbol{\beta},\rho)$
and $(\boldsymbol{\alpha}-x,
\boldsymbol{\beta}+x,\rho)$
are identical for all $x\in\mathbb{R}$.
Thus, without loss of generality,
we assume $\alpha_{1}^* =\hat{\alpha}_1=
 0$.
For the exponential model
without interaction effect
(i.e., $\rho=0$),
\cite{yan2015asymptotics} verified the uniform consistency
and \cite{yan2016asymptotics} even proved the
asymptotic normality of the MLE by analysing
the inverse of the sample covariance matrix.
The following
theorem \ref{asymth1} verifies
the "weak consistency" of
the MLE in model (\ref{asymmodel1}).

\begin{theorem}\label{asymth1}
Suppose $\alpha_i^*,\beta_i^*,\rho^*$ are uniformly
bounded within $[-B,B]$.
Set $\Delta \alpha=
\frac{1}{n}\sum\limits_{i\leq n}
( \hat{\alpha}_i-\alpha_{i}^*
),
\Delta \beta=\frac{1}{n}
\sum\limits_{i\leq n}
( \hat{\beta}_i-\beta_{i}^*)$.
Then we have:
\begin{itemize}
\item
$(\hat{\rho}-\rho)^2=O_p(n^{-\frac{1}{2}}
(\log n)^2)$;

\item
$
(\Delta \alpha+\Delta \beta)^2 =
O_p(n^{-\frac{1}{2}}
(\log n)^2);$
\item
$||\boldsymbol
{\hat{\alpha}}-
\boldsymbol{\alpha}^*-
\Delta\alpha||_2^2
+||\boldsymbol
{\hat{\beta}}-
\boldsymbol{\beta}^*-
\Delta\beta||_2^2
=O_p(n^{\frac{1}{2}}
(\log n)^2)
.$

\end{itemize}

\end{theorem}

As we show in the proof in appendix,
uniform consistency can
be obtained from the "weak
consistency".

\begin{theorem}\label{asymth3}
Suppose $\alpha_i^*,\beta_i^*,\rho^*$ are uniformly
bounded within $[-B,B]$. Then we have:
$$||\hat{\boldsymbol{\alpha}}
-\boldsymbol{\alpha}^*
||_{\infty}+
||\hat{\boldsymbol{\beta}}
-\boldsymbol{\beta}^*||_\infty
=O_p(n^{-\frac{1}{8}}\sqrt{\log n}).$$
\end{theorem}

Both theorem \ref{asymth1},\ref{asymth3}
are conclusions of more
 general theorems given in subsection
 \ref{asymsubsec2}.

\subsection{A general result on uniform consistency}
\label{asymsubsec2}
We make the following regular
and bounding assumption on
$f$ and $\mathcal{B}_i$.
\begin{assumption}\label{asymass1}

\ \\

\begin{enumerate}
\item
There exists a compact
subset of $\mathbb{R}^d$,
$\mathcal{B}$, with
$\mathcal{B}_i\subseteq \mathcal{B}$
for all $i\leq n$ and
 $\boldsymbol{\rho}\in \mathcal{B}$.

\item For some constant $0<c\leq C$,
$c\leq f(y;\boldsymbol{\al}_i,\boldsymbol{\al}_j,\boldsymbol{\rho}) \leq C$
for all $\boldsymbol{\al}_i\in\mathcal{B}_i,
\boldsymbol{\al}_j\in\mathcal{B}_j,\boldsymbol{\rho}\in \mathcal{B},
y\in\mathcal{S}-\mb{0}$.

\item There exists
a constant $C_3$
such that for any $\b{\eta}',\b{\rho}'$,
$\b{\eta}'',\b{\rho}''$:
$$
\sum\limits_{i<j\leq n}
||f(\cdot;\b{\eta}_i',\b{\eta}_j',\b{\rho}')
-
f(\cdot;\b{\eta}_i'',\b{\eta}_j'',\b{\rho}'')
||_2^2
\geq
C_3\cdot \inf\limits_{\b{\al},\b{\rho\cong \b{\al}'',\b{\rho}''}}
\big\{\ ||\boldsymbol{\al}'-\boldsymbol{\al}||_2^2
+n||\boldsymbol{\rho}'-\boldsymbol{\rho}||_2^2
\ \big\}.
$$

\item For some constant $C_1>0$,
$
\big|\big|\frac{\partial f(y,\boldsymbol{\al}_i,\boldsymbol{\al}_j,\boldsymbol{\rho})}
{\partial \boldsymbol{\al}_i}\big|\big|_2,
\big|\big|
\frac{\partial f(y,\boldsymbol{\al}_i,\boldsymbol{\al}_j,\boldsymbol{\rho})}
{\partial \boldsymbol{\al}_j}\big|\big|_2$,
$
\big|\big|\frac{\partial f(y,\boldsymbol{\al}_i,\boldsymbol{\al}_j,\boldsymbol{\rho})}
{\partial \boldsymbol{\rho}}\big|\big|_2$
,
$
\big|\big|\frac{\partial^2 f(y,\boldsymbol{\al}_i,\boldsymbol{\al}_j,\boldsymbol{\rho})}
{\partial \boldsymbol{\al}_i^2}\big|\big|_2$,
$\big|\big|
\frac{\partial^2
f(y,\boldsymbol{\al}_i,\boldsymbol{\al}_j,\boldsymbol{\rho})}
{\partial \boldsymbol{\al}_j^2}\big|\big|_2$,
$\big|\big|\frac{\partial^2
 f(y,\boldsymbol{\al}_i,\boldsymbol{\al}_j,\boldsymbol{\rho})}
{\partial \boldsymbol{\rho}^2}\big|\big|_2
<C_1
$
for all $\boldsymbol{\al}_i\in\mathcal{B}_i,
\boldsymbol{\al}_j\in\mathcal{B}_j,\boldsymbol{\rho}\in \mathcal{B}$,
$y\in \mathcal{S}-\{\mb{0}\}$.
Here the second derivative
of $f$ is regarded as a $d^2$ length vector.

\item There exists $C_2$ such that for any
any true parameter
$\b{\al}^*,\b{\rho}^*$, any
$i\leq n$, any $\b{\al}_i\in\mathcal{B}_i$,
$$
\sum\limits_{j\ne i}
||f(\cdot;
\boldsymbol{\al}_i,\boldsymbol{\al}_{j}^*,\boldsymbol{\rho}^*)
-
f(\cdot;
\boldsymbol{\al}_i^*,\boldsymbol{\al}_{j}^*,\boldsymbol{\rho}^*)
||_2
\geq C_2
n \cdot
||\boldsymbol{\al}_i
-\boldsymbol{\al}_i^*||_2^2
.$$
\end{enumerate}
\end{assumption}

Assumption $5$
 means that the latent type of an individual
has an significant impact on the profile
of connection probability.
Assumption $3$ is identification assumption.
In \citep{zhao2012consistency} it is required
that the criteria (which is  log likelihood
function in this paper) is uniquely maximized
over the space of parameters.
These two assumptions together resemble the
assumption (c) and (*)
 in \citep{zhao2012consistency}
section 4.
In many papers, latent type $\al$ is
random.
In these settings,
it is not necessary that assumption
$3$ and $5$ are satisfied with probability 1.
But it is usually satisfied with large
probability.
Assumption 2 reflects a balance property 
of the network. It means that there is no
node that has too large or too small
degree. Technically, this
 assumption guarantee that
the criteria we used, namely the log likelihood
function, is Lipschitz continuous in
$\b{\eta},\b{\rho}$. So this assumption
resembles assumption (a) in
\citep{zhao2012consistency}
section 4.
Assumption $1,4$ are regular assumptions.
Such assumptions appeared frequently in the
literature (see also assumption (b) in
 \citep{zhao2012consistency}
section 4).

\begin{theorem}\label{asymth5}
Suppose the model satisfy
assumption \ref{asymass1} item (1)(2).
If
$(\log n)^2=o(n\mu_n^*)$
then there exists true parameter
$\b{\al}^*,\b{\rho}^*$ such that:
\begin{align}
\sum\limits_{i<j\leq n}
||f(\cdot;\hat{\boldsymbol{\al}}_i,\hat{\boldsymbol{\al}}_j,\hat{\boldsymbol{\rho}})
- f(\cdot;\boldsymbol{\al}_i^*,\boldsymbol{\al}_j^*,\boldsymbol{\rho}^*)||_2^2
=O_p(n^{\frac{3}{2}}(\mu_n^*)^{-\frac{1}{2}}(\log n)^2).
\end{align}
\end{theorem}
Here $f(\cdot; \boldsymbol{\al}_i,\boldsymbol{\al}_j,
\boldsymbol{\rho})$ is
regarded as a vector in $[0,1]^{|\mathcal{S}-\{\mb{0}\}|}$.
Theorem \ref{asymth5}
says that the MLE well
fits the profile of connection probability
if $(\log n)^4=o(n\mu_n)$, 
 and the
average prediction error
 is of order
$\frac{(\log n)^2}{\sqrt{n\mu_n}}$.
While it is also important to
know to what extent can the
network data predict the latent type. 
In many model settings,
the latent type $\boldsymbol{\al}$
is not just an auxiliary variable.
Inferring the latent type
from the network data, known as
community detection, is an important task for its
own sake.
The following corollary establishes the
weak consistency of the MLE in model (\ref{asymmodel2})
and follows directly
from theorem \ref{asymth5}.
\begin{corollary}\label{asymcol2}
Assume the model satisfies assumption
 \ref{asymass1} item (1)(2)(3). If
$(\log n)^2=o(n\mu_n^* )$
then we have:
\begin{align}
\inf\limits_{\b{\al},\b{\rho} \text{ is true}}
\big\{\ ||\hat{\boldsymbol{\al}}-\boldsymbol{\al}||_2^2
+n||\hat{\boldsymbol{\rho}}-\boldsymbol{\rho}||_2^2
\ \big\}
=O_p( n^{\frac{1}{2}}(\mu_n^*)^{-\frac{1}{2}}
(\log n)^2).
\end{align}
\end{corollary}
Corollary \ref{asymcol2} verifies
the weak consistency of the MLE under the condition
that 1) the model
does not suffer from serious model
identification problem. i.e., if
assumption \ref{asymass1} item (3)
is satisfied;
and 2) $(\log n)^4=o(\mu_n n)$. 
Note that the condition
ensuring the weak consistency, namely $(\log n)^4=o(\mu_n n)$,
is slightly stronger than that in the discrete
parameter
case, namely
$1 = o(\mu_n n)$
\citep{zhao2012consistency}.
As we show in the following theorem \ref{asymth4},
there is a standard way to transform the
$||\cdot||_2$ error bound into the uniform
error bound.
\begin{theorem}\label{asymth4}
Assume the model satisfies assumption
\ref{asymass1} item (2)(4)(5).
For any
 true parameter $\boldsymbol{\al}^*,\boldsymbol{\rho}^*$,
let $\sum\limits_{i\leq n}
||\hat{\boldsymbol{\al}}_i-\boldsymbol{\al}_i^*||_2^2
=\Delta_\al$,
$||\hat{\boldsymbol{\rho}}-\boldsymbol{\rho}^*||_2^2
 = \Delta_\rho$.
Then we have,

\begin{align}
\sup\limits_i
||\hat{\boldsymbol{\al}_i}-
\boldsymbol{\al}_i^*||_2^2
=O_p(n^{-\frac{1}{2}}
\Delta_\al^{\frac{1}{2}}
(\mu_n^*)^{-1}\log n
+\sqrt{\Delta_\rho}\log n
+n^{-\frac{1}{2}}(\mu_n^*)^{-\frac{1}{2}}
(\log n)^{\frac{3}{2}}).
\end{align}

\end{theorem}
We emphasize that theorem \ref{asymth4}
take advantage of $\hat{\boldsymbol{\al}},
\hat{\boldsymbol{\rho}}$ being a MLE. Many
community detection methods estimate the latent type
by maximizing some criterion. The criterion in
the MLE  is the log likelihood.
Our method can be generalized to establish
consistency of criterions other than
log likelihood (see also
\cite{zhao2012consistency} section 2). %
Combine
corollary \ref{asymcol2} with
theorem \ref{asymth4} we
obtain the following uniform consistency condition.
\begin{corollary}\label{asymcol1}
  Assume the model satisfies assumption
\ref{asymass1} item (1)-(5).
  If $(\log n)^2=o(n\mu_n^* )$ then
  we have:
  $$
 \inf\limits_{\b{\al}\text{ is true}}
 \big\{\
  \sup\limits_{i}
  ||\hat{\boldsymbol{\al}}_i-
  \boldsymbol{\al}_i
  ||_2^2
  \ \big\}
   =
   O_p(n^{-\frac{1}{4}}(\mu_n^*)^{-\frac{5}{4}}
   (\log n)^2).
  $$
  Thus if $\frac{n^{\frac{1}{5}}\mu_n^*}{
  (\log n)^{\frac{8}{5}}}\rightarrow\infty$
  then the MLE is uniformly consistent.
\end{corollary}

Corollary \ref{asymcol1} follows directly
from theorem \ref{asymth4} and
 corollary \ref{asymcol2}. Its proof is
 omitted.
The uniform consistency condition we obtained
in the continuous parameter case
 is much stronger than
 that in the discrete parameter case. 
 Namely,  $n^{\frac{1}{5}}\mu_n/(\log n)^{\frac{8}{5}}\rightarrow\infty$
(corollary \ref{asymcol1}) versus
 $n\mu_n/\log n\rightarrow\infty$
  \citep{zhao2012consistency}.
It is not known whether the condition
is necessarily stronger. Our guess is yes.
 However, none of our experiments confirm this.
  On the other hand, we do discover an overfitting phenomenon
 of the MLE in an experiment. 
The uniform  error bound can be reduced
when the stopping criterion of the optimization
algorithm is relaxed.
 Therefore we propose the following open question.
 \begin{question}
 Is the MLE in model (\ref{asymmodel2})
 of $\b{\eta}$ uniformly consistent
 if $(\log n)^K = o(\mu_n n)$? Where
 $K$ is some constant.
 \end{question}
Intuitively,
 continuous parameter space
 is nothing but the limit of
 discrete parameter space.
 Therefore it is natural to discretize
  the parameter space into sufficiently small (but not too small)
pieces
and apply the MLE on the
  discretized parameter space.
We show in
 the next subsection
 that by doing so,
  uniform consistency condition becomes almost
  the same as that in the discrete parameter space.

\subsection{Improve uniform consistency by discretization}
\label{asymsubsec5}
A discretization
 of $\mathcal{B}_1\times \cdots\times \mathcal{B}_n\times\mathcal{B}$
is simply
a subset $\mtB_1\times\cdots\times\mtB_n\times\mathcal{B}
$ of
 $\mathcal{B}_1\times \cdots\times \mathcal{B}_n\times\mathcal{B}$
where each $\mtB_i$ is a finite set.
The discretized parameter space
 $\mtB_1\times\cdots\times\mtB_n\times\mathcal{B}
$ may
include none of the true parameters.
In this section we show that uniform consistency of the following
algorithm holds (corollary \ref{asymcol4}):
 first discretize $\mathcal{B}_1\times \cdots\times \mathcal{B}_n\times\mathcal{B}$
  with appropriately fineness (see (\ref{asymfine})),
then perform MLE over the discretized parameter space.
We define the \emph{pseudo true}
parameter as
 the maximizer of the expected
 log likelihood function
 over $\mtB_1\times\cdots\times\mtB_n\times
 \mathcal{B}$,
 denoted by
 $\b{\t{\al}}^*,\b{\t{\rho}}^*$.
The following quantities characterize
the fineness of the discretized parameter
space $\mtB_1\times\cdots\times\mtB_n\times\mathcal{B}
$:
\begin{align}\label{asymfine}
 &\delta_n' =
 \sup\limits_{i\leq n,
 \b{\al}_i}
 \ \ \inf\limits_
 {\b{\t{\al}}_i\in \t{\mathcal{B}}_i}
 ||\b{\al}_i-\b{\t{\al}}_i||_2
 \\ \nonumber
 &\delta_n = \inf\big\{\
 ||\b{\t{\al}}_i-\b{\t{\al}}_i'||_2:
  \b{\t{\al}}_i\ne\b{\t{\al}}_i'
  \in\mathcal{\t{B}}_i,i\leq n\big\}
.\end{align}
We allow the fineness to scale
with $n$.
Assume that the discretized
parameter space $\t{\mathcal{B}}_i$ is
finite and
$\delta_n',\delta_n>0$.
We first show that the
 pseudo true parameter is closed to the
 true parameter.

 \begin{proposition}\label{asymprop2}
  \begin{enumerate}
 \item
 Assume the model satisfies
 assumption \ref{asymass1} item
(2)(3)(4).
We have, there exists a
  constant $C_4>0$,
  for  any pseudo
 true parameter $\b{\t{\al}}^*,\b{\t{\rho}}^*$,
 there exists a true parameter
 $\b{\al}^*,\b{\rho}^*$ such that,
 $$
||\b{\t{\al}}^*-
 \b{\al}^*||_2^2
 \leq C_4 n\delta_n'^2
  ,\
 ||\b{\t{\rho}}^*-\b{
 \rho}^*||_2^2
 \leq C_4\delta_n'^2
.$$

\item Assume the model satisfies
 assumption \ref{asymass1} item
(4)(5).
There exists a
  constant $C_5>0$ such that
  for
  any pseudo
 true parameter $\b{\t{\al}}^*,\b{\t{\rho}}^*$,
 any true parameter
 $\b{\al}^*,\b{\rho}^*$,
 let $\Delta_\al = ||\b{\t{\al}}^*-
 \b{\al}^*||_2^2, \Delta_\rho =
||\b{\t{\rho}}^*-
 \b{\rho}^*||_2^2 $, we have:
$$
\sup\limits_i
 ||\b{\t{\al}}_i^*-
 \b{\al}_i^*
 ||_2^2 \leq C_5(\delta_n'^2+
 \frac{\Delta_\al}{n}+\Delta_\rho)
.$$

  \end{enumerate}

  Thus if the model satisfies assumption
  item (2)-(5), then there exists
   some constant
  $C_6$, such that for any pseudo
 true parameter $\b{\t{\al}}^*,\b{\t{\rho}}^*$,
 there exists a true parameter $\b{\al}^*$
with:
  $$
\sup\limits_i
 ||\b{\t{\al}}_i^*-
 \b{\al}_i^*
 ||_2^2 \leq C_6\delta_n'^2
.$$
 \end{proposition}

Consider the MLE restricted on the discretized
parameter space, i.e.,
the maximizer of the likelihood function
over the discretized parameter space,
namely $\b{\hat{\t{\al}}},
\b{\h{\t{\rho}}}$. It is easy to see that
item (1)-(4) of
assumption \ref{asymass1} hold for
the discretized parameter space
 if they hold for the original
parameter space. So theorem
\ref{asymth5} and corollary
\ref{asymcol2} holds for the discretized parameter
space with true parameter replaced by pseudo
true parameter.
Therefore, using corollary \ref{asymcol2} and proposition
\ref{asymprop2} conclusion (2),
 the following
weak consistency result follows directly
and its proof is omitted.
\begin{corollary}\label{asymcol3}
Assume the model satisfies assumption \ref{asymass1}
item
(1)-(4).
If $(\log n)^2 = o(n\mu_n^*)$  then
we have:
\begin{align}\nonumber
  &
  \inf\limits_{\b{\al},
  \b{\rho} \text{ is true}}
  \big\{\ ||\bhtal-
  \b{\al}  ||_2^2
  + n||\b{\h{\t{\rho}}}-
   \b{\rho}||_2^2
\  \big\}
   \ = \
   O_p(n^{\frac{1}{2}}(\mu_n^*)^{-\frac{1}{2}}
   (\log n)^2)+
   O(n\delta_n'^2).
   \end{align}
Thus if, in addition,
$(\log n)^4=o(n\mu_n)$ and $\delta_n'\rightarrow 0$ the
MLE on the discretized parameter space is
weak consistent.

\end{corollary}
 The $||\cdot||_2$ error bound of the MLE
 on the discretized parameter space is
 worse than that on the original space.
In contrast, in what follows,
we show that the uniform error bound is
 reduced in the discretized parameter space.

\begin{theorem}\label{asymth6}
Assume the model satisfies assumption
\ref{asymass1} item  (1)-(4)
and item (5) with true parameter replaced
by pseudo true parameter.
Let
$s = |\mathcal{S}|$.
If for a given constant $\varepsilon>0$,
 it holds that
 with probability at least
$1-\varepsilon$ there exists
pseudo true parameter
$ \btal^*,\b{\t{\rho}}^*$
with
$$
||\bhtal-\btal^*||_2^2<\conditionofdeltathird
,\ \
  \log n\sqrt{\frac{||\bhtal-\btal^*||_2^2}{n}}
  +\log n\sqrt{||\b{\h{\t{\rho}}}-
  \b{\t{\rho}}^*||_2^2}+
    \frac{(\log n)^{\frac{3}{2}}}
{\sqrt{n\mu_n^*}}<\frac{\delta_n}{2C_\varepsilon}
$$,
where $C_\varepsilon$ is the constant in
lemma \ref{asymlem4},
then with probability at least $1-\varepsilon-\frac{1}{n^2}$
there exists pseudo true parameter
$\btal^*$ with
$ \bhtal=\btal^*$.

\end{theorem}
By theorem \ref{asymth6} and proposition
\ref{asymprop2},
uniform consistent estimator
can be obtained by
first discretizing the parameter
space and then applying the
MLE on the discretized parameter space. Since
the discretization is  artificial and
the pseudo true parameter is closed to
the true parameter,
assumption \ref{asymass1} item 5
is likely to hold for the discretized parameter
space.
Furthermore, theorem \ref{asymth5}, \ref{asymth6},
proposition \ref{asymprop2} together
give the optimal discretization parameter
$\delta_n,\delta_n'$, i.e.,
$\delta_n'=\Theta(\delta_n)$,
$\delta_n =\Theta\big( n^{-\frac{1}{6}}(\mu_n^*)^{-\frac{1}{6}}
(\log n)^{\frac{4}{3}}\big)$
.
\begin{corollary}\label{asymcol4}
Assume the model satisfies
 assumption \ref{asymass1}
(1)-(4)
and item (5) with true parameter replaced
by pseudo true parameter.
Assume the discretization satisfies:
$\delta_n'=\Theta(\delta_n)$,
$\delta_n =\Theta\big( n^{-\frac{1}{6}}(\mu_n^*)^{-\frac{1}{6}}
(\log n)^{\frac{4}{3}}\big)$.
If $(\log n)^2 = o(n\mu_n^*)$,
then
 we have:
\begin{align}
\inf\limits_{\b{\al} \text{ is true }}
\big\{\ \sup\limits_i\
||\bhtal_i-\b{\al}||_2^2
\ \big\}
=O_p(n^{-\frac{1}{3}}(\mu_n^*)^{-\frac{1}{3}}
(\log n)^{\frac{8}{3}}).
\end{align}
Thus uniform consistency of the MLE on the discretized
parameter space
in model (\ref{asymmodel2}) holds
if $(\log n)^8 = o(n\mu_n^*)$ and
$\delta_n'=\Theta(\delta_n)$,
$\delta_n =\Theta\big(
 n^{-\frac{1}{6}}(\mu_n^*)^{-\frac{1}{6}}
(\log n)^{\frac{4}{3}}
\big)$.
\end{corollary}
The proof of corollary \ref{asymcol4}
follows directly from
theorem \ref{asymth5}, \ref{asymth6},
proposition \ref{asymprop2}
and is therefore
 omitted.

\section{Simulation}
\label{asymsecsimulation}

In subsection \ref{asymsecsimulation}
the experiments verify that
the MLE is asymptotically normal in the
exponential network model with interaction
effect when $\mu_n = 1$.
In subsection \ref{asymsubsec7},
we check if the uniform consistency holds under
weaker condition than that in corollary \ref{asymcol1}.
The experiment results incline a positive
answer.
In subsection \ref{asymsubsec8}, we
show that the MLE can be overfitting.
Both  uniform error bound and
mean square error bound  are not reduced when
the stopping criterion,
which we used in the optimization algorithm approximating the MLE,
 is strengthen.

We employ the exponential network model
with interaction effect in all experiments.
To approximate the MLE, we use a
combination of gradient descent (or Newton-Raphson)
 and coordinate descent.
We iteratively update the estimated parameter $\b{\eta},\b{\rho}$
until a certain stopping criterion is satisfied.
In each round of iteration, we
update parameters $\b{\al}_i,i=1,\cdots,n$ and $\b{\rho}$
 one by one (coordinate
descent). To update  $\b{\al}_i$,
we either use
gradient descent to maximize (with respect to $\b{\al}_i$)
the log likelihood, or use  Newton-Raphson
to find a zero of the partial derivative of
the log likelihood
(with respect to $\b{\al}_i$). Both methods
yield a sub-iteration. The sub-iteration is
stopped  until certain sub-stopping-criterion is satisfied.
The details
 are given in each subsection.

\subsection{Asymptotic normality}
\label{asymsecsimulation}
In this experiment,
we check the asysmptotic normality of MLE in model
(\ref{asymmodel1}).
We use the following algorithm to approximate
the maximizer of the likelihood:
\begin{algorithm}[H]
\caption{Coordinate descent and Newton-Raphson}
\begin{algorithmic}
\State $\alpha_i^{(0)} = \alpha_i^*;\beta_i^{(0)} = \beta_i^*;\rho^{(0)} = \rho^*$;$t=0$;
\While{ $e^{(t)}<10^{-3}$}
\For{ $i=1,2,\cdots,n+1$}
  \If{ $i\leq n$}
\State approximate the zero of $\frac{\partial \ell_i}{\partial \alpha_i},\frac{\partial \ell_i}{\partial \beta_i}$ by Newton-Raphson method. i.e.,
\State iteratively update $\alpha_i,\beta_i$ as in Newton-Raphson method until the following sub-stopping-criterion is
\State reached,
\State $e_{i,s}= \max\{\ |\alpha_i^{s+1}-\alpha_i^{s}|, |\beta_i^{s+1}-\beta_i^s|\ \}<10^{-4}$.
 \State Where $\alpha_i^s,\beta_i^s$ refers to the updated  $\alpha_i,\beta_i$ after the $s^{th}$  iteration.
\Else { $i=n+1$}
 \State update $\rho$ in the same way.
 \EndIf
\EndFor

\State Let $\b{\alpha}^{(t)},\b{\beta}^{(t)},\rho^{(t)}$  be the parameter vector after the $t^{th}$ round iteration.
\State Let  $e^{(t)} = \sup\limits_{i\leq n} \big\{\ |\alpha_i^{(t)}-\alpha_i^{(t-1)}|,|\beta_i^{(t)}-\beta_i^{(t-1)}|, |\rho^{(t)}-\rho^{(t-1)}|\ \big\}$.
\State t = t+1;
 \EndWhile

 \end{algorithmic}
\end{algorithm}

We conduct three groups of experiments
corresponding to
figure \ref{asymsimul1},\ref{asymsimul2}
and \ref{asymsimul3} respectively.
The parameters  in the
three groups are generated in the
following way.
\begin{enumerate}
\item $\alpha_1 = 0,\mu_n=1$ (is known); $\alpha_i (i\ne 1),\beta_i$ are
independent uniformly distributed over [-1,1]
 and $\rho=0.6$.

 \item $\alpha_1=0,\mu_n=1$ (is known); $\alpha_i (i\ne 1),\beta_i$ are
 independent standard normal  and
 $\rho = 0.3$.

 \item $\alpha_1=0,\mu_n=1$ (is known); $\alpha_i\in
 \{0.3,0.7\} (i\ne 1)$
 and $\beta_i\in \{0.4, 0.6\}$ are
 independent and chosen uniformly at random; $\rho = -0.7$.
\end{enumerate}
For each group of parameter, we generate
networks of size $n=100,200,400$ by model
(\ref{asymmodel1}). In each setting 200 replications
are used.
The parameters
and the initial value of parameters in the optimization
algorithm are fixed through out all replications.
 We plot the Q-Q plot
comparing the empirical distribution
of the following random variables
 against the standard normal distribution:
 $\sqrt{n}(\h{\alpha}_2-\alpha^*_2)$,
$\sqrt{n}(\h{\alpha}_{50}-\alpha^*_{50})$, $\sqrt{n}(\h{\alpha}_{100}-\alpha^*_{100})$,
 $\sqrt{n}(\h{\beta}_1-\beta^*_1)$, $\sqrt{n}(\h{\beta}_{50}-\beta^*_{50})$,
 $\sqrt{n}(\h{\beta}_{100}-\beta^*_{100})$ and
 $n(\h{\rho}-\rho^*)$.
 The parameters $\alpha_1,\cdots,\alpha_{100},
 \beta_1,\cdots,\beta_{100}$ are fixed throughout all
 three groups of experiments in order to make the comparison.
 The results show that
 the distribution
 of all these random variables in all three groups of experiments
 converge to some normal distribution.

\begin{figure}[H]
\centering
\includegraphics[width=4.5in,height=3.5in]{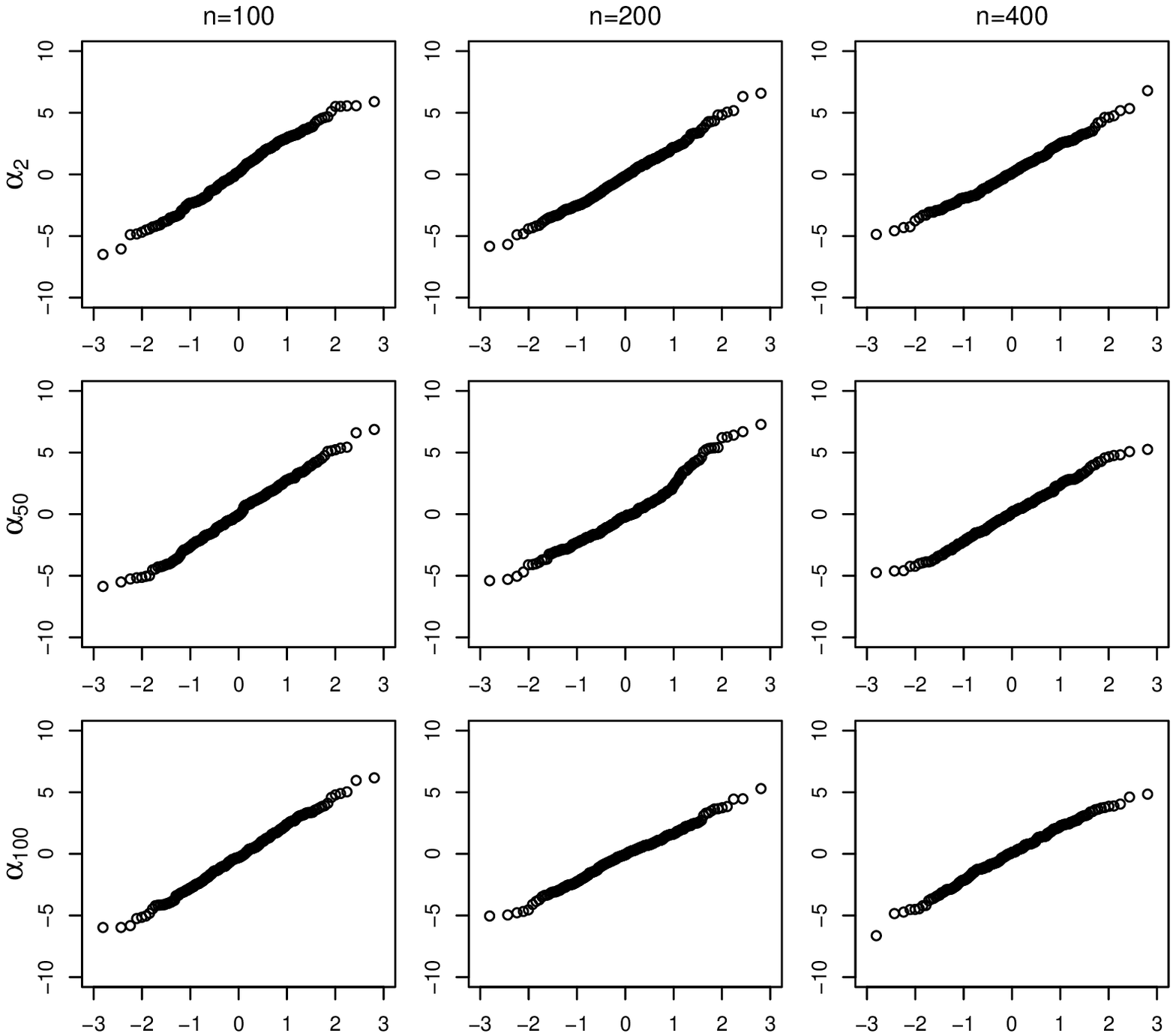} \\
\includegraphics[width=4.5in,height=3.5in]{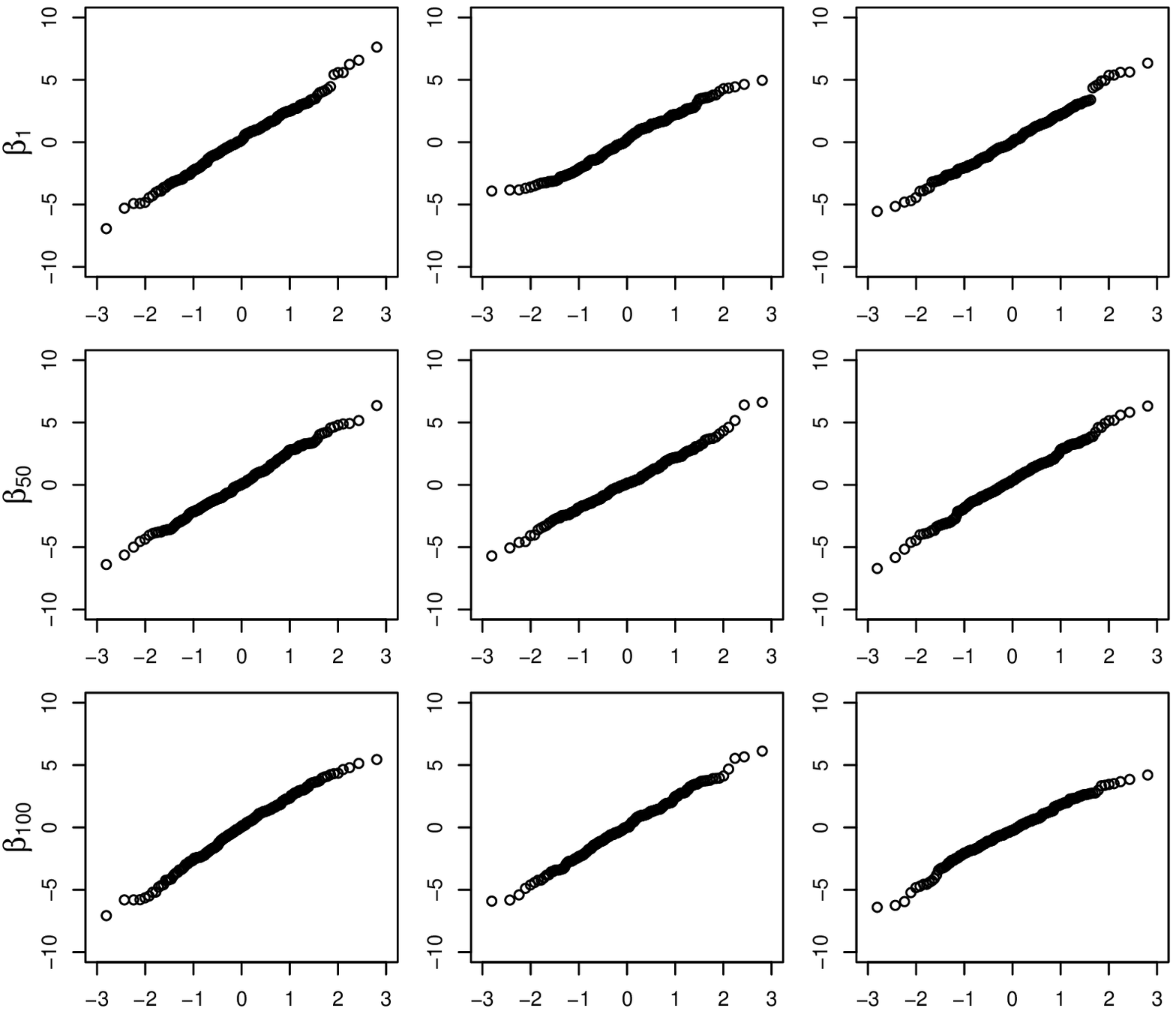} \\
\includegraphics[width=4.5in,height=1.2in]{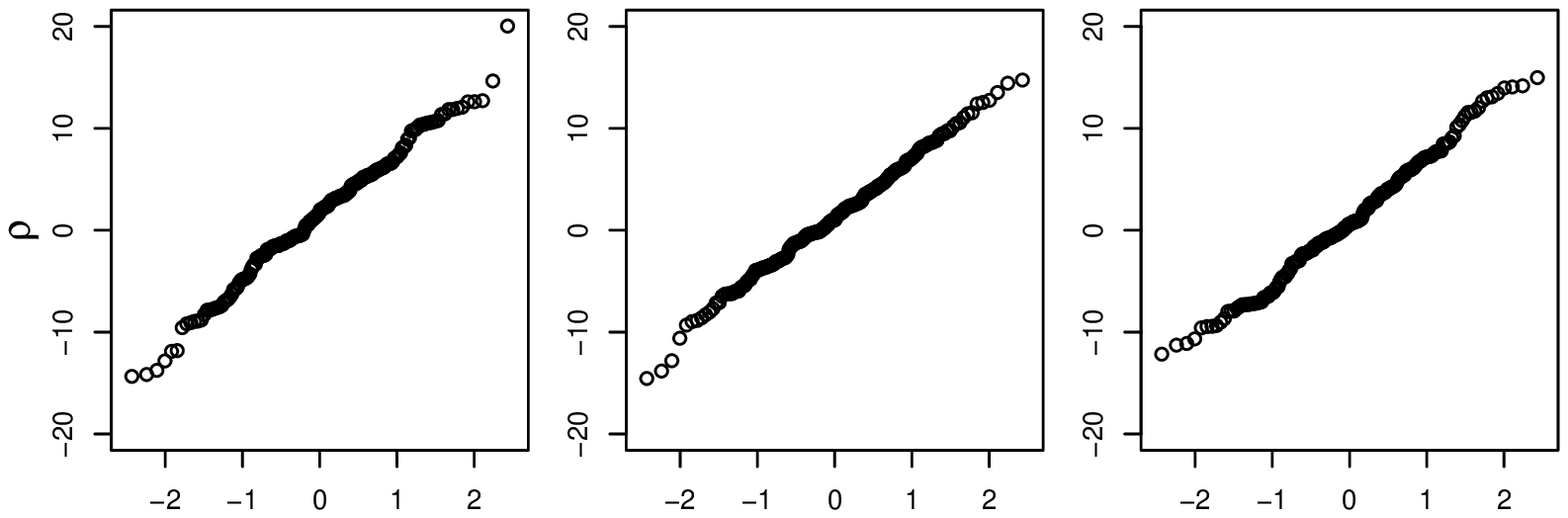}
\caption{The Q-Q plots of the  error (group 1).}
\label{asymsimul1}
\end{figure}

\begin{figure}[H]
\centering
\includegraphics[width=4.5in,height=3.5in]{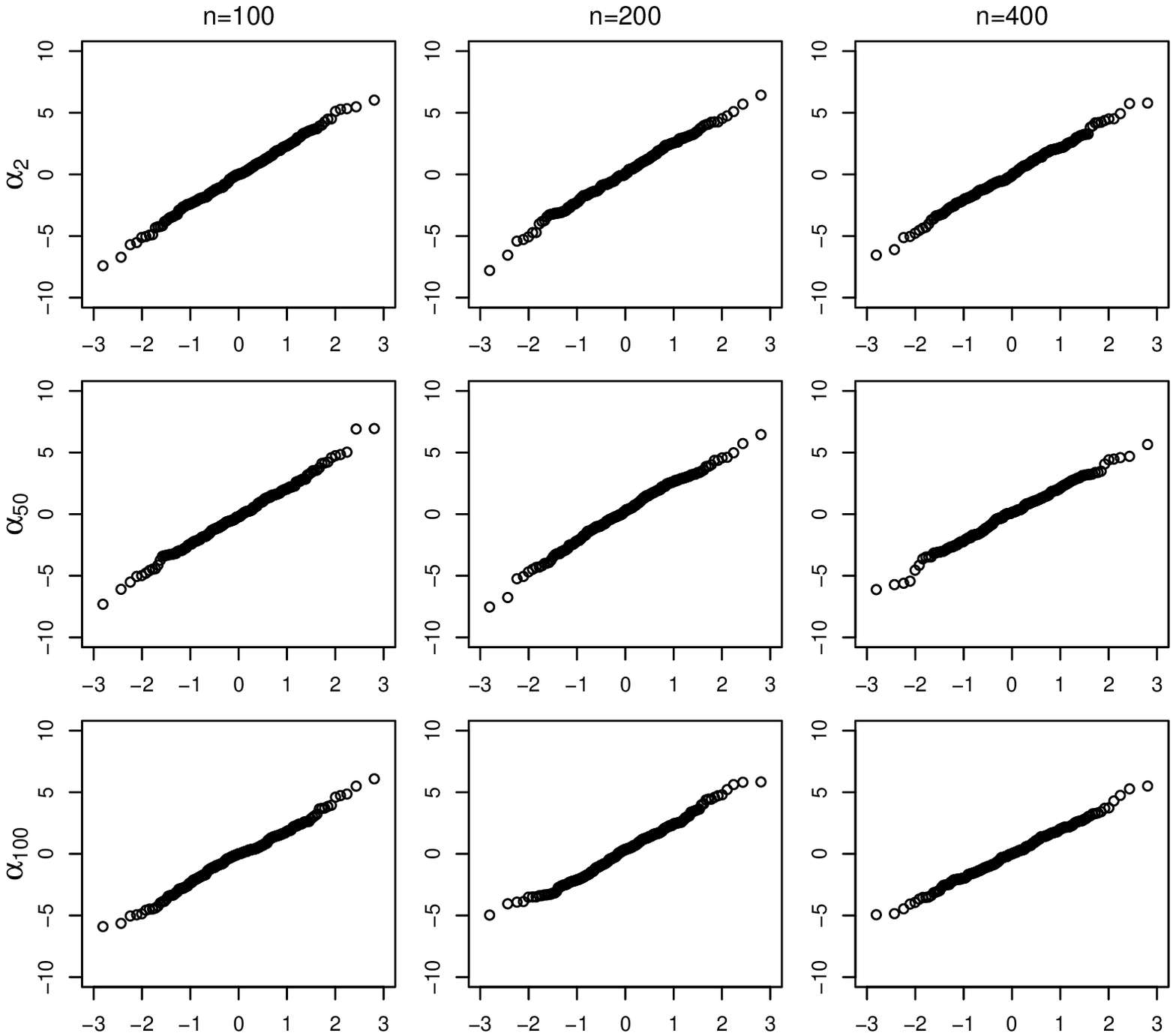} \\
\includegraphics[width=4.5in,height=3.5in]{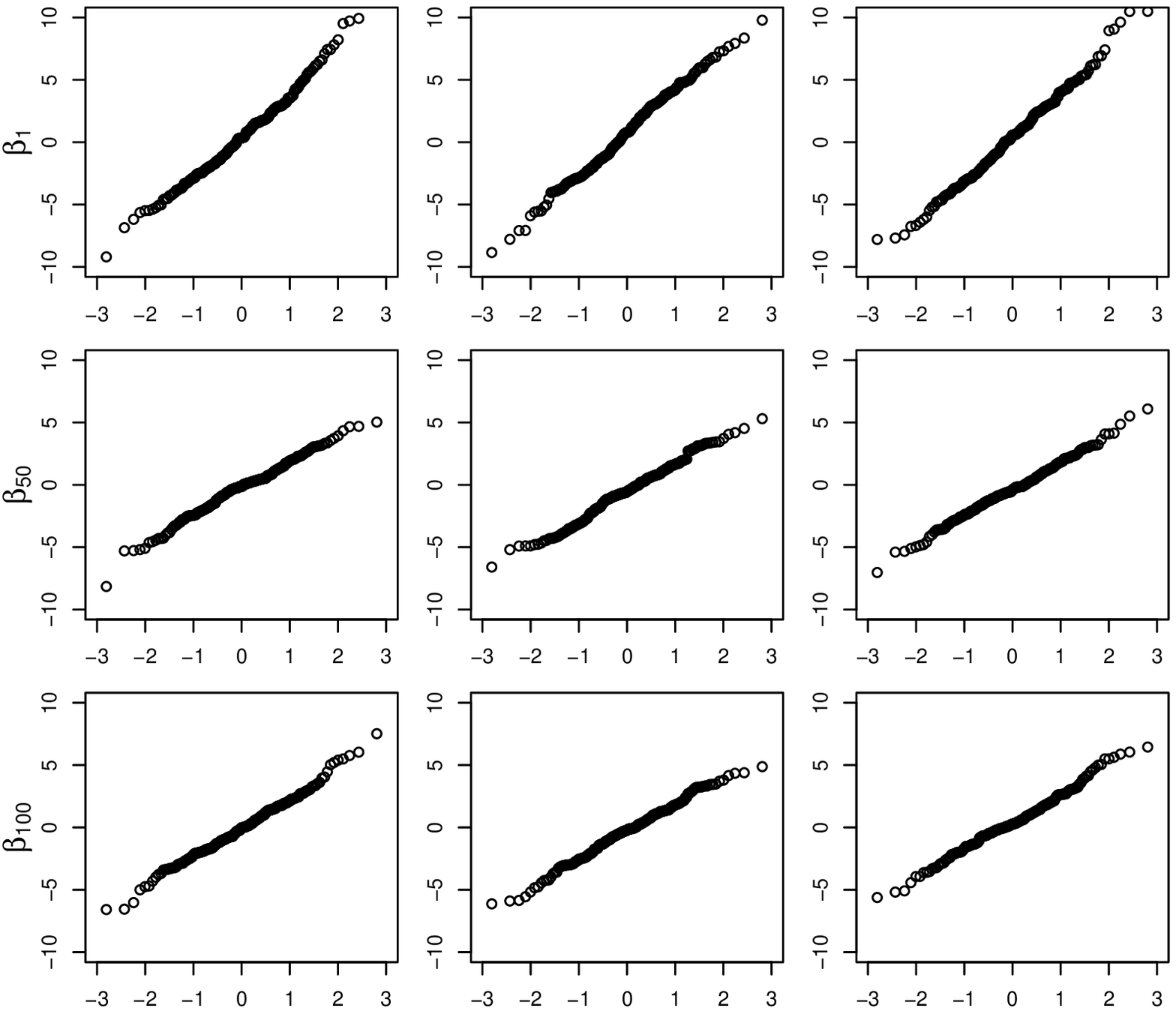} \\
\includegraphics[width=4.5in,height=1.2in]{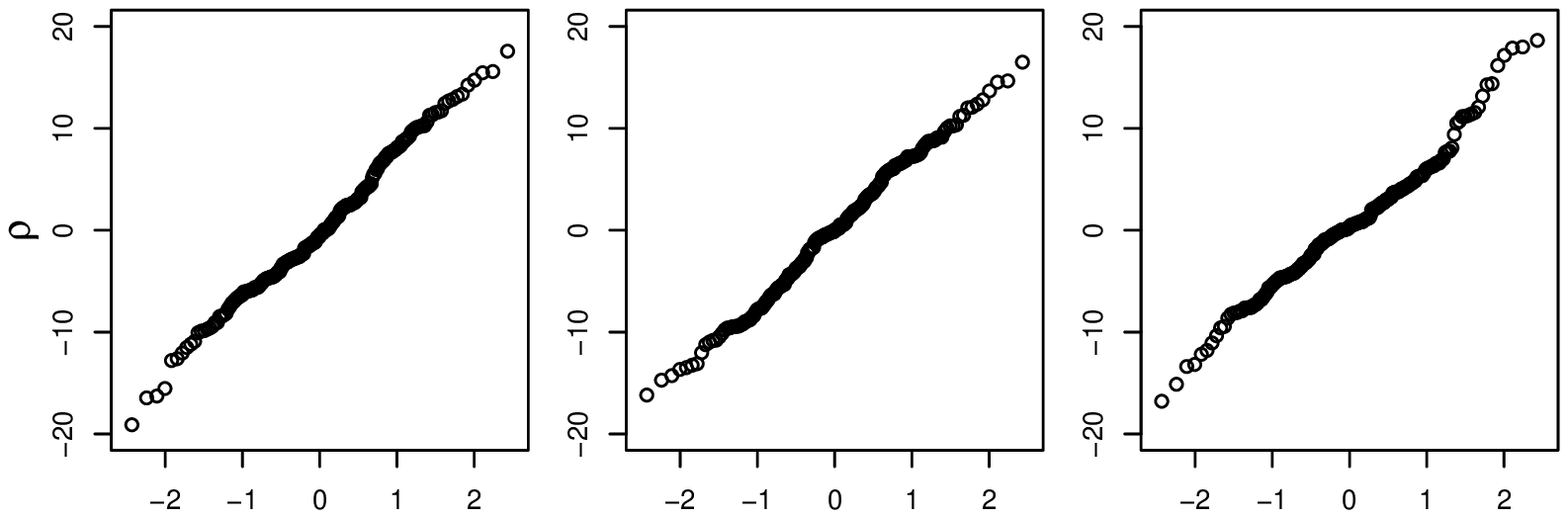}
\caption{The Q-Q plots of the error (group 2).}
\label{asymsimul2}
\end{figure}

\begin{figure}[H]
\centering
\includegraphics[width=4.5in,height=3.5in]{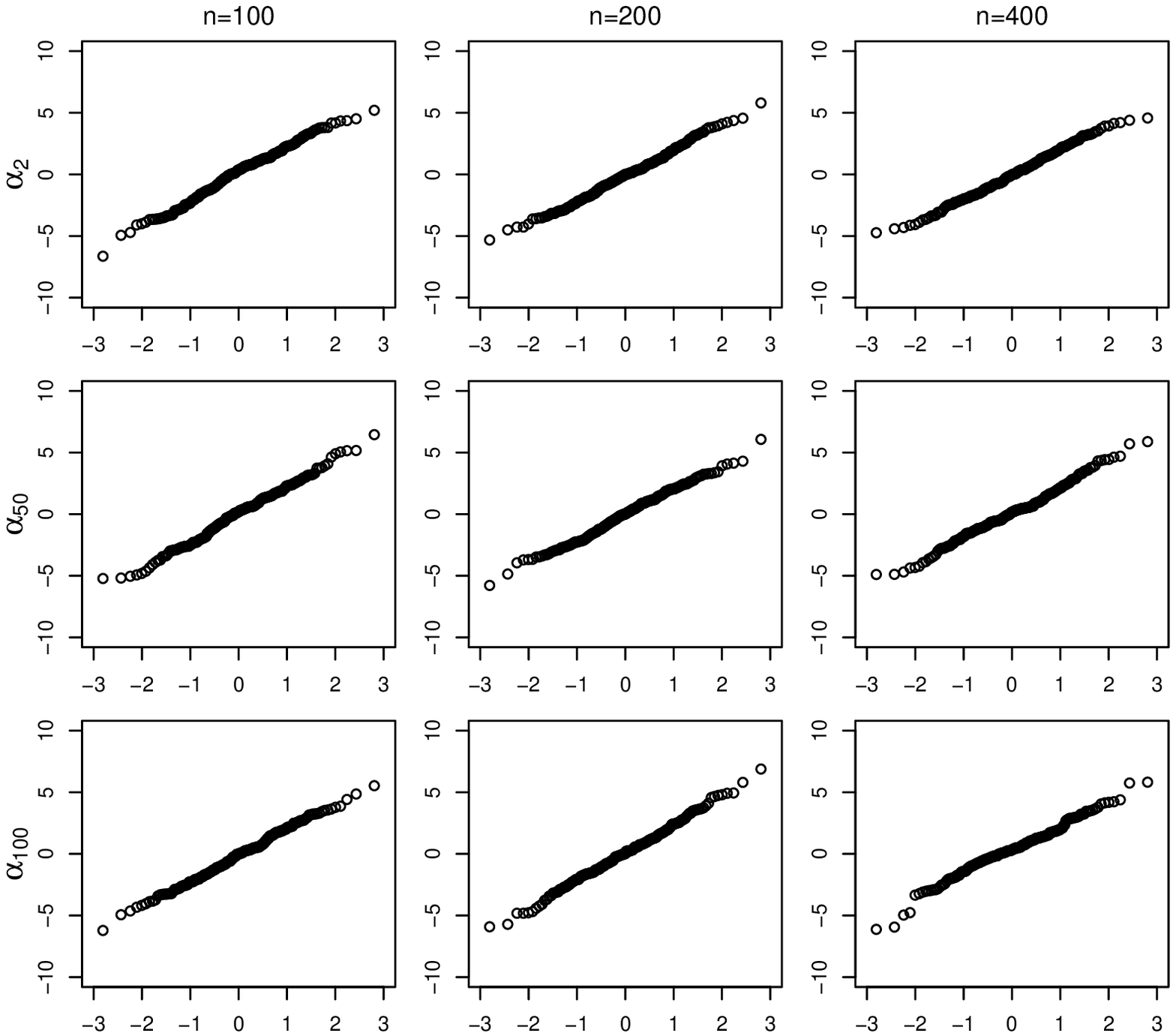} \\
\includegraphics[width=4.5in,height=3.5in]{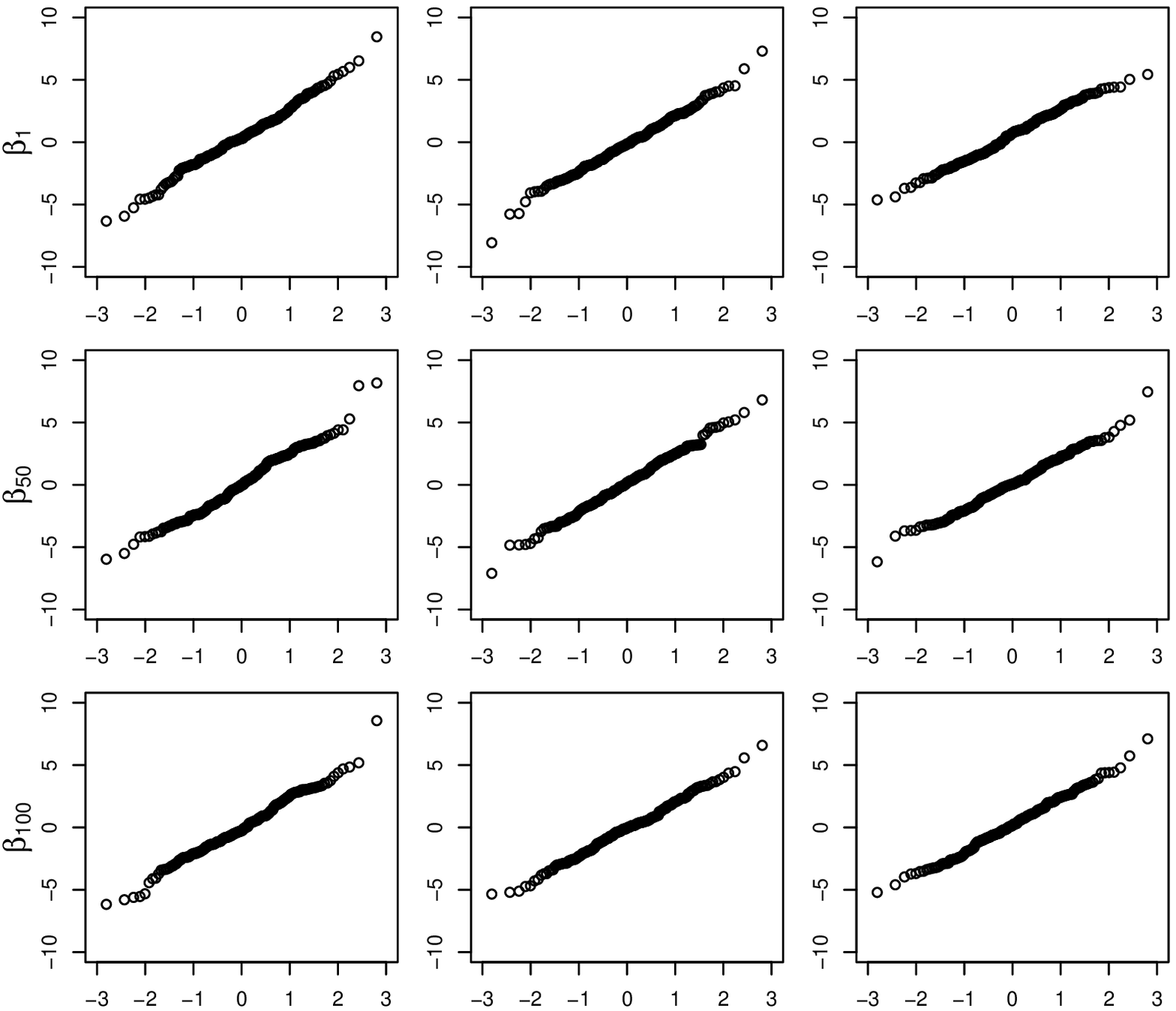} \\
\includegraphics[width=4.5in,height=1.2in]{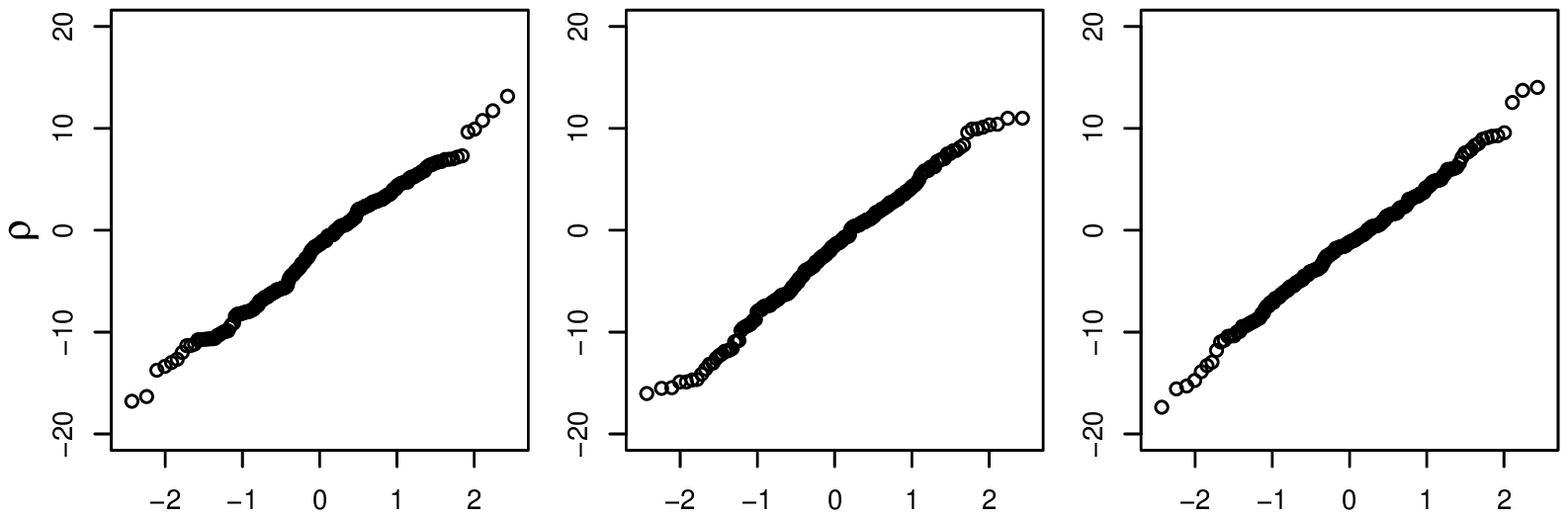}
\caption{The Q-Q plots of the error (group 3).}
\label{asymsimul3}
\end{figure}

\subsection{Uniform consistency condition}
\label{asymsubsec7}
The purpose of this subsection is to check if uniform
consistency holds under much weaker condition than
that in corollary \ref{asymcol1}.
We conduct one group of
experiment, figure \ref{asymsimul4}, on  model
(\ref{asymmodel1}) with
 $\mu_n=10\cdot n^{-\frac{2}{3}}$ being known,
 $\alpha_1 = 0$; $\alpha_i (i\ne 1), \beta_i$ being
 independently uniformly distributed over [-1,1]
 and $\rho=0.3$.
We use the following algorithm to approximate
the maximizer of the likelihood:
\begin{algorithm}[H]
\caption{Coordinate descent and gradient descent}
\begin{algorithmic}
\State  $\alpha_i^{(0)} = \alpha_i^*;\beta_i^{(0)} = \beta_i^*;\rho^{(0)} = \rho^*$;$t=0$;
\While{ $e^{(t)}<0.05$ }
\For{ $i=1,2,\cdots,n+1$}
  \If{ $i\leq n$}
\State maximize $\ell_i$ with respect to $\alpha_i,\beta_i$ by gradient descent algorithm, i.e.,
\State iteratively update $\alpha_i,\beta_i$ as in gradient descent algorithm until the following sub-stopping-criterion is
\State reached:
\State $e_{i,s}= \max\{\ |\frac{\partial \ell_i}{\partial \alpha_i}(\b{\alpha}_{-i}^{(t-1)},\b{\beta}_{-i}^{(t-1)},
\alpha_{i}^{s},\beta_i^s,\rho^{(t-1)},\mu_n)|, |\frac{\partial \ell_i}{\partial \beta_i}(\b{\alpha}_{-i}^{(t-1)},\b{\beta}_{-i}^{(t-1)},
\alpha_{i}^{s},\beta_i^s,\rho^{(t-1)},\mu_n)|\ \}<10^{-3}$.
 \State Where $\alpha_i^s,\beta_i^s$ refers to the updated  $\alpha_i,\beta_i$ after the $s^{th}$  iteration
 and $\alpha_i^0 = \alpha_i^{(t-1)},\beta_i^0 = \beta_i^{(t-1)}$.
\Else { $i=n+1$}
 \State update $\rho$ in the same way.
 \EndIf
\EndFor

\State Let $\b{\alpha}^{(t)},\b{\beta}^{(t)},\rho^{(t)}$  be the parameter vector after the $t^{th}$ round iteration.
\State Let  $e^{(t)} = \sup\limits_{i\leq n} \big\{\
|\frac{\partial \ell_i}{\partial \alpha_i}(\b{\alpha}^{(t)},\b{\beta}^{(t)},\rho^{(t)},\mu_n)|,
|\frac{\partial \ell_i}{\partial \beta_i}(\b{\alpha}^{(t)},\b{\beta}^{(t)},\rho^{(t)},\mu_n)|,
|\frac{\partial \ell}{\partial \rho}(\b{\alpha}^{(t)},\b{\beta}^{(t)},\rho^{(t)},\mu_n)|\ \big\}$.
\State t = t+1;
 \EndWhile

 \end{algorithmic}
\end{algorithm}

We generate
networks of size $n=400,1400$ and $5000$ by model
(\ref{asymmodel1}).
In each setting 30 replications are used.
The parameters
and the initial value of parameters in the optimization
algorithm are fixed through out all replications.
We use mean square error bound, namely
$\frac{1}{n}\sum\limits_{i\leq n}
\big(\ |\hat{\alpha}_i-\alpha_i^*|^2+|\h{\beta}_i-
\beta_i^*|^2\big)$, as the weak consistency measure.
We use uniform error bound, namely
$\sup\limits_{i\leq n}
\big\{\ |\hat{\alpha}_i-\alpha_i^*|^2+ |\h{\beta}_i-
\beta_i^*|^2\ \big\}$,
 as the uniform consistency measure.
The results show that,
both mean square error bound (subfigure \ref{asymsimul4-a})
and the uniform error bound (subfigure \ref{asymsimul4-b})
decreases as $n$ increases.

\begin{figure}[H]
  \centering
  \subfigure[Mean square error bound]{
    \label{asymsimul4-a} 
    \includegraphics[width=2.5in,height=2.5in]{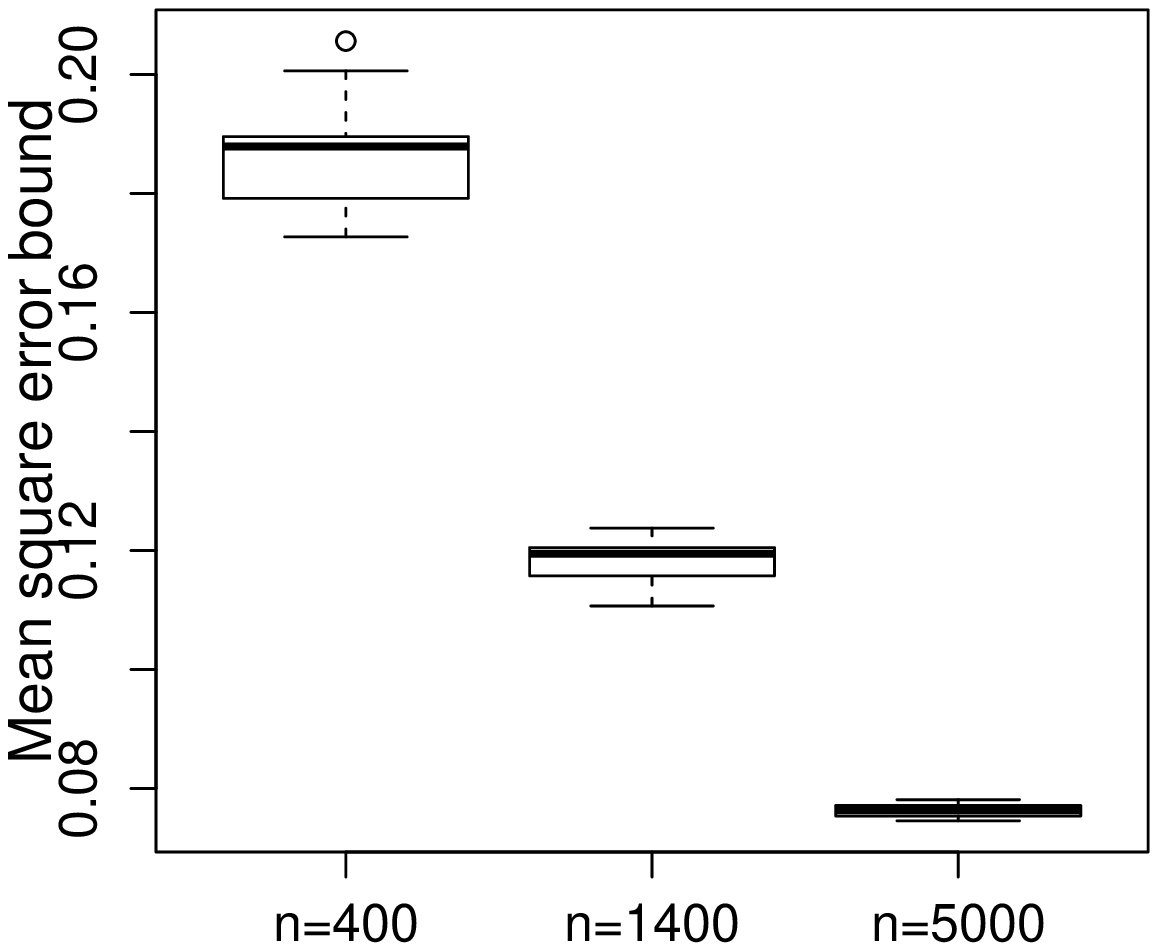}}
  \subfigure[Uniform error bound]{
    \label{asymsimul4-b} 
    \includegraphics[width=2.5in,height=2.5in]{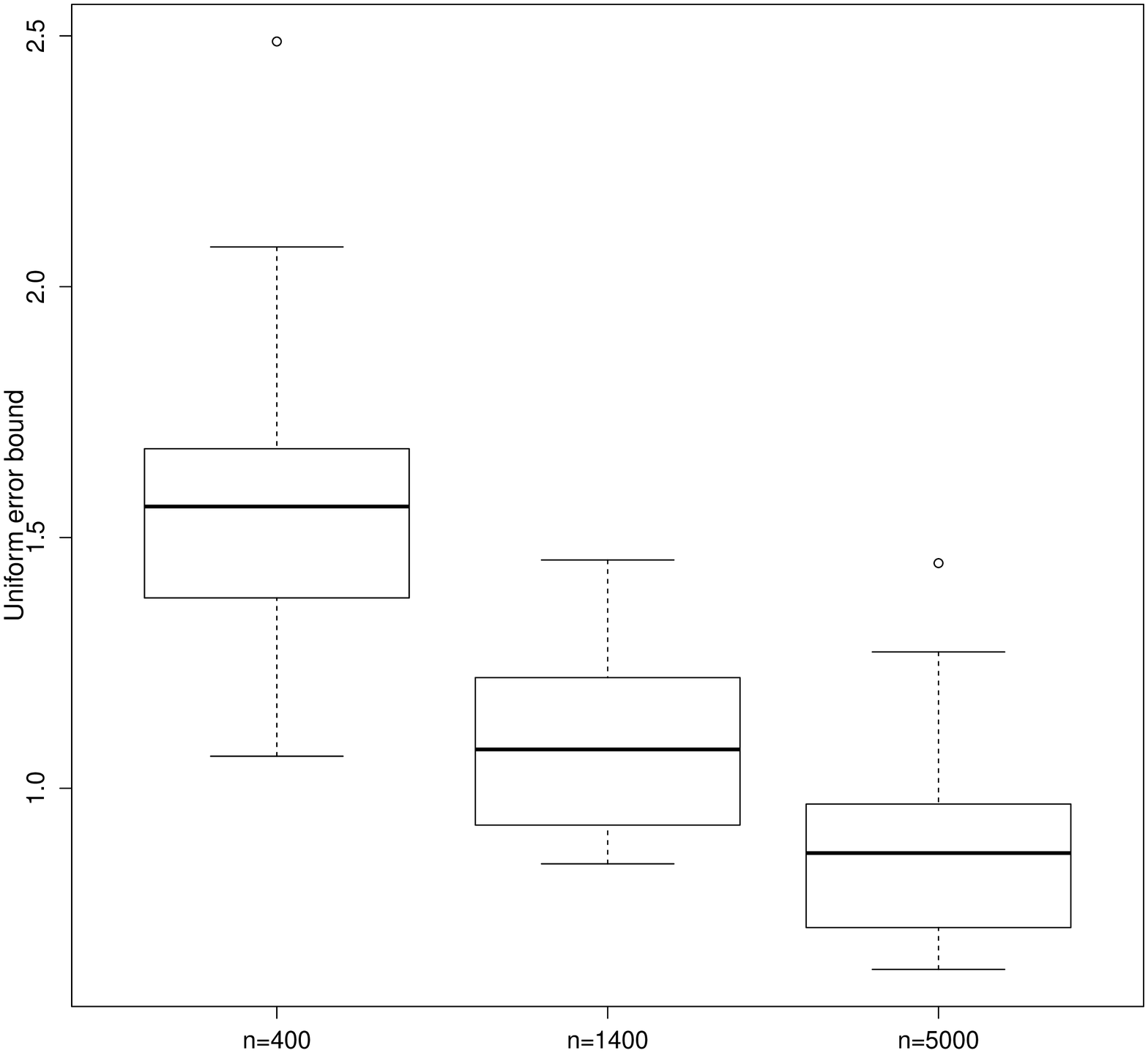}}

  \caption{Consistency under weaker condition}
  \label{asymsimul4} 
\end{figure}

\subsection{Overfitting}
\label{asymsubsec8}

In this subsection, our results show that the MLE can be
overfitting. Because the uniform error bound can be reduced when
the stopping criterion is relaxed. Meanwhile, the mean square
error bound is not reduced.
We conduct one group of
experiment, figure \ref{asymsimul4}, on the exponential network model
(\ref{asymmodel1}) with
 $\mu_n=10\cdot n^{-\frac{2}{3}}$ being known,
 $\alpha_1 = 0$; $\alpha_i (i\ne 1), \beta_i$ being
 independently uniformly distributed over [-1,1]
 and $\rho=0.3$.
We use the following algorithm to approximate
the maximizer of the likelihood. Let $\varepsilon_i,\varepsilon'_i$
be independent and uniformly  distributed over [-0.5,0.5].
\begin{algorithm}[H]
\caption{Coordinate descent and gradient descent}
\begin{algorithmic}
\State  $\alpha_i^{(0)} = \alpha_i^*+\varepsilon_i;\beta_i^{(0)} = \beta_i^*+\varepsilon'_i;$ $\rho^{(0)} = \rho^*$;$t=0$;
\While{ $e^{(t)}<e$ }
\For{ $i=1,2,\cdots,n+1$}
  \If{ $i\leq n$}
\State maximize $\ell_i$ with respect to $\alpha_i,\beta_i$ by gradient descent algorithm, i.e.,
\State iteratively update $\alpha_i,\beta_i$ as in gradient descent algorithm until the following sub-stopping-criterion is
\State reached:
\State $e_{i,s}= \max\{\ |\frac{\partial \ell_i}{\partial \alpha_i}(\b{\alpha}_{-i}^{(t-1)},\b{\beta}_{-i}^{(t-1)},
\alpha_{i}^{s},\beta_i^s,\rho^{(t-1)},\mu_n)|, |\frac{\partial \ell_i}{\partial \beta_i}(\b{\alpha}_{-i}^{(t-1)},\b{\beta}_{-i}^{(t-1)},
\alpha_{i}^{s},\beta_i^s,\rho^{(t-1)},\mu_n)|\ \}<10^{-3}$.
 \State Where $\alpha_i^s,\beta_i^s$ refers to the updated  $\alpha_i,\beta_i$ after the $s^{th}$  iteration
 and $\alpha_i^0 = \alpha_i^{(t-1)},\beta_i^0 = \beta_i^{(t-1)}$.
\Else { $i=n+1$}
 \State update $\rho$ in the same way.
 \EndIf
\EndFor

\State Let $\b{\alpha}^{(t)},\b{\beta}^{(t)},\rho^{(t)}$  be the parameter vector after the $t^{th}$ round iteration.
\State Let  $e^{(t)} = \sup\limits_{i\leq n} \big\{\
|\frac{\partial \ell_i}{\partial \alpha_i}(\b{\alpha}^{(t)},\b{\beta}^{(t)},\rho^{(t)},\mu_n)|,
|\frac{\partial \ell_i}{\partial \beta_i}(\b{\alpha}^{(t)},\b{\beta}^{(t)},\rho^{(t)},\mu_n)|,
|\frac{\partial \ell}{\partial \rho}(\b{\alpha}^{(t)},\b{\beta}^{(t)},\rho^{(t)},\mu_n)|\ \big\}$.
\State t = t+1;
 \EndWhile

 \end{algorithmic}
\end{algorithm}
We refer to the constant $e$ in
the above algorithm as the stopping criterion.
We generate
 networks of size $n=1000,2000$ by model
(\ref{asymmodel1}) and
test the algorithm for $e= 0.08, 0.02$.
In each setting 30 replications are used.
The parameters
and the initial value of parameters in the optimization
algorithm are fixed through out all replications.
Figure \ref{asymsimul5} shows how mean square error bound
(subfigure \ref{asymsimul5-a})
 and uniform error bound (subfigure \ref{asymsimul5-b})
vary with stopping criterion $e$.
 Both error bounds are not explicitly
 reduced when $e$ decreases.

\begin{figure}[H]
  \centering
  \subfigure[Mean square error bound]{
    \label{asymsimul5-a} 
    \includegraphics[width=4in,height=2in]{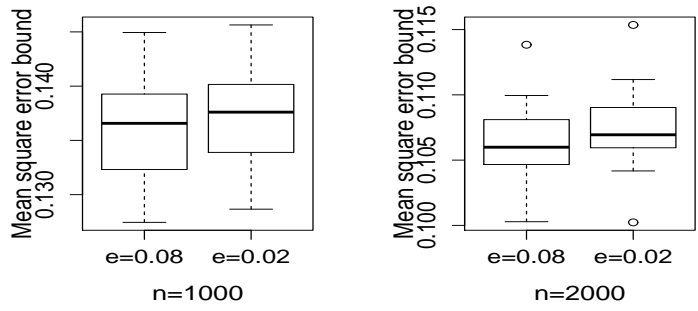}}
  \hspace{0.3in}
  \subfigure[Uniform error bound]{
    \label{asymsimul5-b} 
    \includegraphics[width=4in,height=2in]{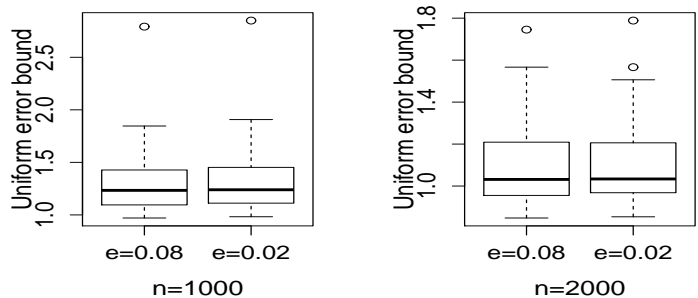}}

  \caption{Overfitting}
  \label{asymsimul5} 
\end{figure}

\section{Conclusion}
Our results confirm that the condition ensuring
weak consistency in continuous parameter case
is almost the same as that in discrete parameter
case. But the uniform consistency condition
we obtained is stronger than that in
the discrete parameter case. It is not clear if the
uniform consistency condition
in the continuous parameter case is necessarily
stronger than that in the discrete parameter case.
Our experiments in section
 \ref{asymsecsimulation} show that the MLE in the exponential
network model with interaction effect might be
asymptotically normal when $\mu_n=1$. This is stronger
than the consistency results we proved.
The experiments in section \ref{asymsubsec7}
 also show that it is possible that
uniform consistency condition holds under much weaker
condition than that in corollary \ref{asymcol1}.
However, the experiment in section
\ref{asymsubsec8} shows that relaxing the stopping criterion
 somehow reduces uniform error bound. This indicates that in
continuous parameter case, MLE may indeed overfit.
We take this as an indirect evidence confirming that
the uniform consistency condition in continuous parameter
case is necessarily stronger (but may not be necessarily
as strong as that in corollary \ref{asymcol1}).
We also show that by applying
the MLE on the discretized parameter space,
the  uniform consistency condition
is almost the same as that in the
discrete parameter case.

\section{Proof}

\subsection{Proof of theorem \ref{asymth5}}

The proof of theorem \ref{asymth5}
concern bounding
the fluctuation of log likelihood
on the whole parameter space.
We firstly establish an auxiliary
lemma on the fluctuation of
the sum of independent discrete
random variable.
\begin{lemma}\label{asymlem3}
Suppose $Z_i,i\leq m$ are
$m$ independent random variables with
$|Z_i|\leq C_m$.
There is a univeral
constant $c_0$ such that
for any $u$, if
$\dfrac{|u|}{\sum\limits_{i\leq m}
\oE[Z_i^2]}<\frac{c_0}{C_m}$,
let
$Z = \sum\limits_{i\leq m} Z_i$,
then we have,
\begin{align}
\mbbP\bigg(
\big|
Z- \oE
[Z]\big|\geq u
\bigg)
\leq
\exp\bigg\{
-\frac{u^2}{2\sum\limits_{i\leq m}
\oE[Z_i^2]}
+\frac
{u^3 \sum\limits_{i\leq m }
\oE[|Z_i|^3]
}{\big(\sum\limits_{i\leq m}
\oE[Z_i^2]\big)^3}
\bigg\}.
\end{align}

\end{lemma}
\begin{proof}
The proof concern Chernoff inequality:
\begin{align}\label{asymeq24}
\mbbP\bigg(
Z- \oE
[Z]\leq -u
\bigg)
\leq \inf\limits_{t<0}
\exp
\big\{
\log(\oE[e^{tZ}])-
\oE[tZ]
+tu\big\}.
\end{align}
Note that
$\log(\oE[e^{tZ}]) =
\sum\limits_{i\leq m}
\log \oE[e^{tZ_i}]$.
Therefore if $|t| <\frac{c_0}{C_m}$, by
Taylor expansion
\begin{align}\label{asymeq23}
\log \oE[e^{tZ_i}]
\leq t\oE[Z_i]+
\frac{1}{2}t^2\oE[Z_i^2]
+ |t|^3 \oE[|Z_i|^3].
\end{align}
Set $t = -\dfrac{u}{\sum\limits_{i\leq m}
\oE[Z_i^2]}$,
so $|t| <\frac{c_0}{C_m}$.
Therefore
substitute (\ref{asymeq23})
into (\ref{asymeq24})
we have,
\begin{align}
\mbbP\bigg(
Z- \oE
[Z]\leq -u
\bigg)
\leq
\exp\big\{
-\dfrac{u^2}{2\sum\limits_{i\leq m}
\oE[Z_i^2]}
+\dfrac
{u^3 \sum\limits_{i\leq m }
\oE[|Z_i|^3]
}{\big(\sum\limits_{i\leq m}
\oE[Z_i^2]\big)^3}
\big\}.
\end{align}

It is similar to prove
 the other direction.

\end{proof}

\begin{lemma}\label{asymlem5}
\begin{align}
\sup\limits_{\boldsymbol{\al}\in\mathcal{B}_1
\times\cdots\times\mathcal{B}_n,
\boldsymbol{\rho}\in\mathcal{B},\mu_n
\in[\frac{1}{n},1]} \bigg|\ell(\mb{Y};\boldsymbol{\al},\boldsymbol{\rho},\mu_n)
-\oE\big[
 \ell(\mb{Y};\boldsymbol{\al},\boldsymbol{\rho},\mu_n)\
\big]\bigg|
=O_p(n^{\frac{3}{2}}\sqrt{\mu_n^*}(\log n)^2)
.\end{align}
\end{lemma}
\begin{proof}
Since $\mathcal{B}$ is compact
by assumption \ref{asymass1} item (1), we cover
$\mathcal{B}_1\times\cdots\times
\mathcal{B}_n$ with
$n^{2n}$ many points,
cover $\mathcal{B}$ with
$n^2$ many points, cover
$[\frac{1}{n},1]$ with
$n^3$ many points
 such that for every
 $\boldsymbol{\al}\in \mathcal{B}_1\times
 \cdots\times \mathcal{B}_n,\boldsymbol{\rho}\in\mathcal{B},
 \mu_n\in[\frac{1}{n},1]$ there exists
 a covering point such that
 the log likelihood on the covering
 point is close to
 that on $\boldsymbol{\al},\boldsymbol{\rho}\in\mathcal{B},
 \mu_n$.
So it suffices to show that
\begin{align}
\sup\limits_{(\boldsymbol{\al}
,\boldsymbol{\rho},\mu_n)\in B_n}
\bigg|\ell(\mb{Y};\boldsymbol{\al},\boldsymbol{\rho},\mu_n)
-\oE\big[
 \ell(\mb{Y};\boldsymbol{\al},\boldsymbol{\rho},\mu_n)\
\big]\bigg|
=O_p(n^{\frac{3}{2}}\sqrt{\mu_n^*}(\log n)^2),
\end{align}
where $B_n$ is the covering set.
By assumption \ref{asymass1} item (2)
on $f$ and
the fact $\mu_n^*\geq \frac{1}{n}$,
we have:
\begin{align}
&|\ell_{ij}(Y_{ij};\boldsymbol{\al}_i,\boldsymbol{\al}_j,\boldsymbol{\rho},\mu_n)|\leq
2\log n,
\\ \nonumber 
&\sum\limits_{i<j\leq n}
\oE\big[|\ell_{ij}(Y_{ij};\boldsymbol{\al}_i,
\boldsymbol{\al}_j,\boldsymbol{\rho},\mu_n)|^3
\big]
=O( \mu_n^* n^2 (\log n)^3),
\\ \nonumber
&\sum\limits_{i<j\leq n}
\oE\big[|\ell_{ij}(Y_{ij};\boldsymbol{\al}_i,
\boldsymbol{\al}_j,\boldsymbol{\rho},\mu_n)|^2
\big]
=\Theta( \mu_n^* n^2 (\log n)^2).
\end{align} 

Set $u = n^{\frac{3}{2}}\sqrt{\mu_n^*}
(\log n)^2$. Clearly, by condition on
$\mu_n^*$ in theorem \ref{asymth5} 
$$\dfrac{u}{\sum\limits_{i<j\leq n}
\oE[|\ell_{ij}
(Y_{ij};\boldsymbol{\al}_i,\boldsymbol{\al}_j,\boldsymbol{\rho},\mu_n)|^2]}
 = O(\frac{1}{\sqrt{\mu_n^* n}})=
 o(\frac{1}{\log n}).$$
 So the condition of lemma \ref{asymlem3}
on $u$ is satisfied.
 Furthermore,
 \begin{align}
 \dfrac{u^3 \sum\limits_{i<j\leq n}
\oE[|\ell_{ij}
(Y_{ij};\boldsymbol{\al}_i,\boldsymbol{\al}_j,\boldsymbol{\rho},\mu_n)|^3]}
{
\big(
\sum\limits_{i<j\leq n}
\oE[|\ell_{ij}
(Y_{ij};\boldsymbol{\al}_i,\boldsymbol{\al}_j,\boldsymbol{\rho},\mu_n)|^2]
\big)^3
}=O(n^{\frac{1}{2}}(\mu_n^*)^{-\frac{1}{2}}
(\log n)^3).
 \end{align}

Therefore by Borel Cantali's lemma and
lemma \ref{asymlem3},
\begin{align}
&
\mbbP
\bigg(
\sup\limits_{
(\boldsymbol{\al},\boldsymbol{\rho},\mu_n)\in B_n}
\bigg|\ell(\mb{Y};\boldsymbol{\al},\boldsymbol{\rho},\mu_n)
-\oE\big[
 \ell(\mb{Y};\boldsymbol{\al},\boldsymbol{\rho},\mu_n)\
\big]\bigg|
\geq u
\bigg)\\ \nonumber
\leq &
\exp
\bigg\{
O(n\log n)-
\dfrac{u^2}{\sum\limits_{i<j\leq n}
\oE[|\ell_{ij}
(Y_{ij};\boldsymbol{\al}_i,\boldsymbol{\al}_j,\boldsymbol{\rho},\mu_n)|^2]}
+
O(n^{\frac{1}{2}}(\mu_n^*)^{-\frac{1}{2}}
(\log n)^3)
\bigg\}
\\ \nonumber
\leq &\exp
\big\{
O(n\log n)- \Omega(n(\log n)^2)
+O(n^{\frac{1}{2}}(\mu_n^*)^{-\frac{1}{2}}
(\log n)^3)
\big\}
\\ \nonumber
&\text{ by condition on }\mu_n^*
\text{ in theorem \ref{asymth5}},
(\log n)^2=o(\mu^*_n n)
\\ \nonumber
\rightarrow&0
\end{align}

\end{proof}
Now it is direct to prove
theorem \ref{asymth5}.
Let $\overline{\mu}_n = \frac{1}{n(n-1)}
\sum\limits_{i\ne j\leq n}
I(Y_{ij}\ne \mb{0})$ if
$\mu_n$ represent average degree;
and for notation consistency,
set $\overline{\mu}_n = \mu_n^*$
 when $\mu_n$ is known.
Calculating
the variance of $\overline{\mu}_n$,
we have
$|\overline{\mu}_n-\mu_n^*| = O_p(\frac{\sqrt{\mu_n^*}}{n})$.
So by assumption \ref{asymass1} item (2),
$$
\big|
\ell(\mb{Y};\b{\al},\b{\rho},\overline{\mu}_n)
-
\ell(\mb{Y};\b{\al},\b{\rho},\mu_n^*)
\big|=O_p(n\sqrt{\mu_n^*})
$$
uniformly in $\b{\al},\b{\rho}$.
Let $\mb{X}$ be an independent copy of $\mb{Y}$.
By lemma \ref{asymlem5},
\begin{align}\label{asymeq49}
&\ell(\mb{Y};\boldsymbol{\al}^*,\boldsymbol{\rho}^*,\mu_n^*)
-\oE\big[
\ell(\mb{Y};\boldsymbol{\al}^*,\boldsymbol{\rho}^*,\mu_n^*)
\big],\\ \nonumber
&\ell(\mb{Y};\boldsymbol{\hat{\al}},\hat{\boldsymbol{\rho}},\mu_n^*)
-\oE_{\mb{X}}\big[
\ell(\mb{X};\boldsymbol{\hat{\al}},\hat{\boldsymbol{\rho}},\mu_n^*)
\big]
\\ \nonumber
=&O_p(n^{\frac{3}{2}}\sqrt{\mu_n^*}(\log n)^2).
\end{align}
Also note that,
\begin{align}\label{asymeq50}
&\big|
\ell(\mb{Y};\boldsymbol{\hat{\al}},\hat{\boldsymbol{\rho}},\overline{\mu}_n)
-
\ell(\mb{Y};\boldsymbol{\hat{\al}},\hat{\boldsymbol{\rho}},\mu_n^*)
\big|
=O_p(n\sqrt{\mu_n^*}),
\\ \nonumber
&
\ell(\mb{Y};\boldsymbol{\hat{\al}},\hat{\boldsymbol{\rho}},\overline{\mu}_n)
\geq
\ell(\mb{Y};\boldsymbol{\al}^*,\boldsymbol{\rho}^*,\overline{\mu}_n),
\text{ since }\boldsymbol{\hat{\al}}\text{ is MLE.}
\\ \nonumber
&\oE_{\mb{X}}\big[
\ell(\mb{X};\boldsymbol{\hat{\al}},\hat{\boldsymbol{\rho}},\mu_n^*)
\big]
\leq
\oE\big[
\ell(\mb{Y};\boldsymbol{\al}^*,\boldsymbol{\rho}^*,\mu_n^*)
\big].
\end{align}
Thus combine (\ref{asymeq49})(\ref{asymeq50}),
\begin{align}\label{asymeq25}
\bigg|\oE_{\mb{X}}\big[
\ell(\mb{X};\boldsymbol{\hat{\al}},\hat{\boldsymbol{\rho}},\mu_n^*)
\big]
-
\oE\big[
\ell(\mb{Y};\boldsymbol{\al}^*,\boldsymbol{\rho}^*,\mu_n^*)
\big]\bigg|
 = O_p(n^{\frac{3}{2}}\sqrt{\mu_n^*}(\log n)^2+n\sqrt{\mu_n^*}).
\end{align}

Meanwhile,
(\ref{asymeq25})
together with the assumption
$c\leq f\leq C$
imply that
\begin{align}
 \mu_n^*\sum\limits_{i<j\leq n}
||f(\cdot;\hat{\boldsymbol{\al}}_i,\hat{\boldsymbol{\al}}_j,\hat{\boldsymbol{\rho}})
- f(\cdot;\boldsymbol{\al}_i^*,\boldsymbol{\al}_j^*,\boldsymbol{\rho}^*)||_2^2
=O_p(n^{\frac{3}{2}}\sqrt{\mu_n^*}(\log n)^2).
\end{align}
So the proof is accomplished.

\subsection{Proof of corollary \ref{asymcol2}}
The conclusion
of theorem \ref{asymth5}
together with assumption \ref{asymass1} item (3)
imply that for some true parameter $\b{\al}^*$,
\begin{align}\nonumber
\mu_n^* n
\sum\limits_{i\leq n}
||\hat{\boldsymbol{\al}}_i-\boldsymbol{\al}_i^*||_2^2
+\mu_n^* n^2
||\hat{\boldsymbol{\rho}}-\boldsymbol{\rho}^*||_2^2
=O_p(n^{\frac{3}{2}}\sqrt{\mu_n^*}(\log n)^2).
\end{align}
Or equivalently,
$$
\sum\limits_{i\leq n}
||\hat{\boldsymbol{\al}}_i-\boldsymbol{\al}_i^*||_2^2
= n^{\frac{1}{2}}(\mu_n^*)^{-\frac{1}{2}}
(\log n)^2,
$$
$$
||\hat{\boldsymbol{\rho}}-\boldsymbol{\rho}^*||_2^2
=n^{-\frac{1}{2}}
(\mu_n^*)^{-\frac{1}{2}}
(\log n)^2.
$$

\subsection{Proof of theorem \ref{asymth4}}

 The key observation is the following:
 \begin{align}\label{asymeq26}
 \sup\limits_{i\leq n,
 \boldsymbol{\al}_i\in\mathcal{B}_i}
 \ell_i(\mb{Y};\boldsymbol{\al}_i,\boldsymbol{\boldsymbol{\al}}_{-i}^*,\boldsymbol{\rho}^*,\mu_n^*)
 -\oE
 \big[
 \ell_i(\mb{Y};\boldsymbol{\al}_i,\boldsymbol{\al}_{-i}^*,\boldsymbol{\rho}^*,\mu_n^*)
 \big]
  = O_p(n^{\frac{1}{2}}
  (\mu_n^*)^{\frac{1}{2}}(\log n)^{\frac{3}{2}}).
 \end{align}
(\ref{asymeq26}) follows by
standard covering method,
concentration inequality
(lemma \ref{asymlem3}) plus
Borel Cantelli's lemma. Here the
space to be covered is $\cup_{i\leq n}
\mathcal{B}_i$, which is substantially
smaller than that
 in the proof of theorem \ref{asymth5}
 (namely
$\mathcal{B}_1\times\cdots\times
\mathcal{B}_n$).

 Another  fact implied by
 assumption \ref{asymass1} item (4) is
  \begin{align}\label{asymeq27}
&\sup\limits_{i}
\bigg|
 \ell_i(\mb{Y};\boldsymbol{\al}_i,\boldsymbol{\al}_{-i}^*,\boldsymbol{\rho}^*,\mu_n^*)
 -
 \ell_i(\mb{Y};\boldsymbol{\al}_i,\boldsymbol{\al}_{-i},\boldsymbol{\rho},
 \mu_n)
 \bigg|
\\ \nonumber
=&
O(||\boldsymbol{\boldsymbol{\al}}-\boldsymbol
 {\al}^*||_{1,\mb{Y}_i})
 +
O_p(
\mu_n^* \log n ||\boldsymbol{\boldsymbol{\al}}-\boldsymbol
 {\al}^*||_{1}+
\mu_n^* n \log n
 ||\boldsymbol{\rho}-\boldsymbol{\rho}^*||_1
 +
 n\log n|\overline{\mu}_n-\mu_n^*|
).
 \end{align}

In notation $||\boldsymbol{\al}-\boldsymbol
 {\al}^*||_{1,\mb{Y}_i}$,
 $\mb{Y}_i$ is regarded as a $0-1$ sequence
 $(I(Y_{i1}\ne \mb{0}), I(Y_{i2}\ne\mb{0}),\cdots,
 I(Y_{in} \ne \mb{0}))$.
To see this note that:
for $y\ne \mb{0}$, by assumption \ref{asymass1}
item (4)
\begin{align}\label{asymeq35}
 &\big|
 \ell_{ij}(y;\boldsymbol{\al}_i,\boldsymbol{\al}_j^*,\boldsymbol{\rho}^*,\mu_n^*)
 -
 \ell_{ij}(y;\boldsymbol{\al}_i,\boldsymbol{\al}_j,\boldsymbol{\rho},
 \mu_n)
 \big|\\ \nonumber
  =& O(||\boldsymbol{\al}_j-\boldsymbol{\al}_j^*||_1+
  ||\boldsymbol{\al}_i-\boldsymbol{\al}_i^*||_1+
 ||\boldsymbol{\rho}-\boldsymbol{\rho}^*||_1
 +
 \frac{1}{\mu_n^*}|\overline{\mu}_n-\mu_n^*|)
\end{align},
 and
 \begin{align}\label{asymeq36} \big|
 &\ell_{ij}(\mb{0};\boldsymbol{\al}_i,\boldsymbol{\al}_j^*,\boldsymbol{\rho}^*,\mu_n^*)
 -
 \ell_{ij}(\mb{0};\boldsymbol{\al}_i,\boldsymbol{\al}_j,\boldsymbol{\rho},
 \mu_n)
 \big|
 \\ \nonumber
  =& O(\mu_n^*||\boldsymbol{\al}_j-\boldsymbol{\al}_j^*||_1+
   \mu_n^*||\boldsymbol{\al}_i-\boldsymbol{\al}_i^*||_1+
 \mu_n^*||\boldsymbol{\rho}-\boldsymbol{\rho}^*||_1
 +
 |\overline{\mu}_n-\mu_n^*|)
 .\end{align}
Equation (\ref{asymeq35}) explains
the term $||\boldsymbol{\boldsymbol{\al}}-\boldsymbol{\al}^*||_{1,
\mb{Y}_i}$ in (\ref{asymeq27}).
The other terms in (\ref{asymeq27})
are explained by (\ref{asymeq36})(\ref{asymeq35})
and the following fact:
\begin{align}\label{asymeq37}
\sup\limits_{i}
\sum\limits_{j\leq n}
I(Y_{ij}=\mb{0})=
O_p(\mu_n^* n\log n).
\end{align}
(\ref{asymeq37}) follow by applying
lemma \ref{asymlem3} plus Borel-Cantelli's
lemma by standard procedure.  

\begin{lemma}\label{asymlem4}
For every $\varepsilon>0$ there exists
$C_\varepsilon$ such that with probability
at least $1-\varepsilon$ the following
holds:
\begin{align}\label{asymeq34}
&
\sup\limits_i\
n\mu_n^*\cdot ||\hat{\boldsymbol{\al}}_i-\boldsymbol{\al}^*_i||_2^2
\ \leq \  C_\varepsilon\big
(\ ||\boldsymbol{\hat{\al}}-\boldsymbol
 {\al}^*||_{1,\mb{Y}_i}
 +\mu_n^*\log n\sqrt{n\Delta_\al}+
 \mu_n^* n\log n\sqrt{\Delta_\rho}
 +n^{\frac{1}{2}}
  (\mu_n^*)^{\frac{1}{2}}(\log n)^{\frac{3}{2}}
  \ \big).
\end{align}
\end{lemma}
\begin{proof}
Using Cauchy-Schwartz  inequality
we have
$\sum\limits_{i\leq n}
||\hat{\boldsymbol{\al}}_i-\boldsymbol{\al}^*_i||_1
 =O( \sqrt{n\Delta_\al})$.
 Also note that
  $n|\overline{\mu}_n-\mu_n^*|
  =O_p(1)$.
Therefore using (\ref{asymeq27}),
we have uniformly in $i,\b{\al}_i$:
\begin{align}\label{asymeq28}
&\big|
 \ell_i(\mb{Y}_i;\boldsymbol{\al}_i,\boldsymbol{\al}_{-i}^*,\boldsymbol{\rho}^*,\mu_n^*)
 -
 \ell_i(\mb{Y}_i;\boldsymbol{\al}_i,\boldsymbol{\hat{\al}}_{-i},\hat{\boldsymbol{\rho}},
 \overline{\mu}_n)
 \big|\\ \nonumber
 =& O(||\boldsymbol{\hat{\al}}-\boldsymbol
 {\al}^*||_{1,\mb{Y}_i})
 +
 O_p(\mu_n^*\log n\sqrt{n\Delta_\al}+
 \mu_n^* n\log n\sqrt{\Delta_\rho}
 +\log n\sqrt{\mu_n^*}).
\end{align}

Note that for all $i\leq n$,
\begin{align}\label{asymeq29}
\ell_i(\mb{Y}_i;
\hat{\boldsymbol{\al}}_i,\boldsymbol{\hat{\al}}_{-i},\hat{\boldsymbol{\rho}},
\overline{\mu}_n)
\geq
\ell_i(\mb{Y}_i;
\boldsymbol{\al}_i^*,\boldsymbol{\hat{\al}}_{-i},\hat{\boldsymbol{\rho}},
\overline{\mu}_n),
\end{align}
since $\hat{\boldsymbol{\al}}$ is
MLE.
While for all $i\leq n$
\begin{align}\label{asymeq30}
\oE_{\mb{X}_i}
\big[
\ell_i(\mb{X}_i;
\hat{\boldsymbol{\al}}_i,\boldsymbol{\hat{\al}}_{-i},\hat{\boldsymbol{\rho}},
\overline{\mu}_n)
\big]
\leq
\oE_{\mb{Y}_i}
\big[
\ell_i(\mb{Y}_i;
\boldsymbol{\al}_i^*,\boldsymbol{\al}_{-i}^*,
\boldsymbol{\rho}^*,\mu_n^*)
\big].
\end{align}
By assumption \ref{asymass1} item (4)
and using Cauchy-Schwarz inequality,
we have, uniformly in $i$:
\begin{align}\label{asymeq51}
\bigg|\ \oE_{\mb{X}_i}
\big[
\ell_i(\mb{X}_i;
\boldsymbol{\al}_i^*,\boldsymbol{\hat{\al}}_{-i},\hat{\boldsymbol{\rho}},
\overline{\mu}_n)
\big]
-
\oE_{\mb{X}_i}
\big[
\ell_i(\mb{X}_i;
\boldsymbol{\al}_i^*,\boldsymbol{\al}_{-i}^*,
\boldsymbol{\rho}^*,\mu_n^*)
\big]\ \bigg|
= O_p(\mu_n^*\log n\sqrt{n\Delta_\al}+
 \mu_n^* n\log n\sqrt{\Delta_\rho}
 +\log n\sqrt{\mu_n^*}
 ).
\end{align}
By (\ref{asymeq30})(\ref{asymeq51}),
uniformly in $i$:
\begin{align}\label{asymeq52}
\oE_{\mb{X}_i}
\big[
\ell_i(\mb{X}_i;
\hat{\boldsymbol{\al}}_i,\boldsymbol{\hat{\al}}_{-i},\hat{\boldsymbol{\rho}},
\overline{\mu}_n)
\big]
\leq
\oE_{\mb{X}_i}
\big[
\ell_i(\mb{X}_i;
\boldsymbol{\al}_i^*,\boldsymbol{\hat{\al}}_{-i},\hat{\boldsymbol{\rho}},
\overline{\mu}_n)
\big]+
O_p(\mu_n^*\log n\sqrt{n\Delta_\al}+
 \mu_n^* n\log n\sqrt{\Delta_\rho}
 +\log n\sqrt{\mu_n^*}
 ).
\end{align}

Combine (\ref{asymeq28}),(\ref{asymeq26})
we have, uniformly in $i$,
\begin{align}\label{asymeq31}
&
\bigg|\ell_i(\mb{Y}_i;
\hat{\boldsymbol{\al}}_i,\boldsymbol{\hat{\al}}_{-i},\hat{\boldsymbol{\rho}},
\overline{\mu}_n)
-
\oE_{\mb{X}_i}
\big[
\ell_i(\mb{X}_i;
\hat{\boldsymbol{\al}}_i,\boldsymbol{\hat{\al}}_{-i},\hat{\boldsymbol{\rho}},
\overline{\mu}_n)
\big]\ \bigg|,
\\ \nonumber
&
\bigg|\ \ell_i(\mb{Y}_i;
\boldsymbol{\al}_i^*,\boldsymbol{\hat{\al}}_{-i},\hat{\boldsymbol{\rho}},
\overline{\mu}_n)
-
\oE_{\mb{X}_i}
\big[
\ell_i(\mb{X}_i;
\boldsymbol{\al}_i^*,\boldsymbol{\hat{\al}}_{-i},\hat{\boldsymbol{\rho}},
\overline{\mu}_n)
\big]\ \bigg|
\\ \nonumber
=&  O(||\boldsymbol{\hat{\boldsymbol{\al}}}-\boldsymbol
 {\al}^*||_{1,\mb{Y}_i})
 +
 O_p(\mu_n^*\log n\sqrt{n\Delta_\al}+
 \mu_n^* n\log n\sqrt{\Delta_\rho}
 )
 +O_p(n^{\frac{1}{2}}
  (\mu_n^*)^{\frac{1}{2}}(\log n)^{\frac{3}{2}}).
\end{align}
Thus combine
(\ref{asymeq29}),(\ref{asymeq52})
with (\ref{asymeq31}) we have,
uniformly in $i$,
\begin{align}\label{asymeq33}
&
\bigg|
\oE_{\mb{X}_i}
\big[
\ell_i(\mb{X}_i;
\hat{\boldsymbol{\al}}_i,\boldsymbol{\hat{\al}}_{-i},\hat{\boldsymbol{\rho}},
\overline{\mu}_n)
\big]\
-
\oE_{\mb{X}_i}
\big[
\ell_i(\mb{X}_i;
\boldsymbol{\al}_i^*,\boldsymbol{\hat{\al}}_{-i},\hat{\boldsymbol{\rho}},
\overline{\mu}_n)
\big]
\bigg|
\\ \nonumber
=&  O(||\boldsymbol{\hat{\al}}-\boldsymbol
 {\al}^*||_{1,\mb{Y}_i})
 +
 O_p(\mu_n^*\log n\sqrt{n\Delta_\al})+
 \mu_n^* n\log n\sqrt{\Delta_\rho}
 )
+O_p(n^{\frac{1}{2}}
  (\mu_n^*)^{\frac{1}{2}}(\log n)^{\frac{3}{2}}).
\end{align}
Using (\ref{asymeq28}) again, (\ref{asymeq33})
becomes: uniformly in $i$,
\begin{align}\label{asymeq34}
&
\bigg|
\oE_{\mb{X}_i}
\big[
\ell_i(\mb{X}_i;
\hat{\boldsymbol{\al}}_i,\boldsymbol{\al}_{-i}^*,\boldsymbol{\rho}^*,
\mu_n^*)
\big]\
-
\oE_{\mb{Y}_i}
\big[
\ell_i(\mb{Y}_i;
\boldsymbol{\al}_i^*,\boldsymbol{\al}_{-i}^*,\boldsymbol{\rho}^*,
\mu_n^*)
\big]
\bigg|
\\ \nonumber
=&  O(||\boldsymbol{\hat{\al}}-\boldsymbol
 {\al}^*||_{1,\mb{Y}_i})
 +
 O_p(\mu_n^*\log n\sqrt{n\Delta_\al})+
 \mu_n^* n\log n\sqrt{\Delta_\rho}
)
+O_p(n^{\frac{1}{2}}
  (\mu_n^*)^{\frac{1}{2}}(\log n)^{\frac{3}{2}}).
\end{align}
But by assumption \ref{asymass1} item (5),
\begin{align}\nonumber
\bigg|
\oE_{\mb{X}_i}
\big[
\ell_i(\mb{X}_i;
\hat{\boldsymbol{\al}}_i,\boldsymbol{\al}_{-i}^*,\boldsymbol{\rho}^*,
\mu_n^*)
\big]\
-
\oE_{\mb{Y}_i}
\big[
\ell_i(\mb{Y}_i;
\boldsymbol{\al}_i^*,\boldsymbol{\al}_{-i}^*,\boldsymbol{\rho}^*,
\mu_n^*)
\big]
\bigg|
=\Omega(
n\mu_n^* \cdot
||\boldsymbol{\hat{\al}}_i
-\boldsymbol{\al}_i^*||_2^2
)
\end{align}
uniformly in $i$.
Thus conclusion of lemma \ref{asymlem4}
follows.
\end{proof}

Using lemma \ref{asymlem4}
with  $||\hat{\boldsymbol{\al}}-
\boldsymbol{\al}||_{1,\mb{Y}_i}$
replaced by $||\hat{\boldsymbol{\al}}-
\boldsymbol{\al}||_1$ and
applying Cauchy-Schwarz inequality, we
easily have:
\begin{align}\nonumber
\sup\limits_i
||\hat{\boldsymbol{\al}_i}-
\boldsymbol{\al}_i^*||_2^2
=O_p(n^{-\frac{1}{2}}
\Delta_\al^{\frac{1}{2}}
(\mu_n^*)^{-1}\log n
+\sqrt{\Delta_\rho}\log n
+n^{-\frac{1}{2}}(\mu_n^*)^{-\frac{1}{2}}
(\log n)^{\frac{3}{2}})
\end{align}
Thus the conclusion of theorem \ref{asymth4}
follows.

 \subsection{Proof of proposition \ref{asymprop2}}
 Note that
 by definition of $\delta_n'$,
  the discretized parameter
 space admit some
 $\b{\t{\al}}$
 with
 $
 ||
 \b{\t{\al}}-
 \b{\al}^*
 ||_2^2 \leq
 n\delta_n'^2
 $ for some true parameter
 $\b{\al}^*$.
  So by assumption
  \ref{asymass1} item (4)
  for some constant
  $C_f>0$ depending only
  on $f$,
  \begin{align}\label{asymeq38}
 0\geq \oE[
  \ell(\mb{Y};
  \b{\t{\al}},\b{\rho}^*,\mu_n^*)]
  -
  \oE[
   \ell(\mb{Y};
  \b{\al}^*,\b{\rho}^*,\mu_n^*)
  ]
  \geq
  -C_f\mu_n^*n^2\delta_n'^2.
   \end{align}
 But
 \begin{align}
 \oE[
  \ell(\mb{Y};
  \b{\t{\al}},\b{\rho}^*,\mu_n^*)]
  \leq
  \oE[
  \ell(\mb{Y};
  \b{\t{\al}}^*,\b{\t{\rho}}^*,\mu_n^*)]
  \leq
  \oE[
   \ell(\mb{Y};
  \b{\al}^*,\b{\rho}^*,\mu_n^*)
  ].
 \end{align}
Therefore by assumption \ref{asymass1} item (2),
 \[
 \sum\limits_{i<j\leq n}
 ||
 f(\cdot;\b{\t{\al}}_i^*,\b{\t{\al}}_j^*,\b{\t{\rho}}^*,
 \mu_n^*)
 -
 f(
 \cdot;
 \b{\al}_i^*,\b{\al}_j^*,\b{\rho}^*,\mu_n^*
 )
 ||_2^2
 \leq
 C_4 n^2\delta_n'^2,
 \]
 for some constant $C_4>0$ and some
 true parameter $\b{\al}^*$.
 Thus the conclusion (1) follows
 from assumption \ref{asymass1} item (3).

To prove  conclusion (2),
  for all $i\leq n$, let
 $\b{\t{\al}}_i'$ be
 the closest point to
 $\b{\al}_i^*$ in
 $\mtB_i$.
So by
assumption \ref{asymass1} item (4) we have,
for some constant $C_5'$
\begin{align}\label{asymeq39}
\sup\limits_{i\leq n}
\bigg| \oE[
  \ell_i(\mb{Y}_i;
  \b{\t{\al}}_i',\b{\t{\al}}_{-i}^*,\b{\t{\rho}}^*,\mu_n^*)]
  -
  \oE[
   \ell_i(\mb{Y}_i;
  \b{\al}_i^*,\b{\al}_{-i}^*,\b{\rho}^*,\mu_n^*)
  ]
  \ \bigg|
\leq C_5'\mu_n^*(n\delta_n'^2+\Delta_\al
+n\Delta_\rho).
\end{align}
But
 \begin{align}\label{asymeq40}
 &\oE[
  \ell_i(\mb{Y}_i;
  \b{\t{\al}}_i',\b{\t{\al}}_{-i}^*,\b{\t{\rho}}^*,\mu_n^*)]
    \leq
    \oE[
  \ell_i(\mb{Y}_i;
  \b{\t{\al}}_i^*,\b{\t{\al}}_{-i}^*,\b{\t{\rho}}^*,\mu_n^*)]
    \leq
   \oE[
   \ell_i(\mb{Y}_i;
  \b{\al}_i^*,\b{\al}_{-i}^*,\b{\rho}^*,\mu_n^*)
  ].
 \end{align}
 Combine (\ref{asymeq39})(\ref{asymeq40}) we have,
 \begin{align}\nonumber
\sup\limits_{i\leq n}
\bigg| \oE[
  \ell_i(\mb{Y}_i;
  \b{\t{\al}}_i^*,\b{\t{\al}}_{-i}^*,\b{\t{\rho}}^*,\mu_n^*)]
  -
  \oE[
   \ell_i(\mb{Y}_i;
  \b{\al}_i^*,\b{\al}_{-i}^*,\b{\rho}^*,\mu_n^*)
  ]
  \ \bigg|
 \leq C_5'\mu_n^*(n\delta_n'^2+
 \Delta_\al+n\Delta_\rho),
 \end{align}
 which implies
 \begin{align}\nonumber
\sup\limits_{i\leq n}
\bigg| \oE[
  \ell_i(\mb{Y}_i;
  \b{\t{\al}}_i^*,\b{\al}_{-i}^*,\b{\rho}^*,\mu_n^*)]
  -
  \oE[
   \ell_i(\mb{Y}_i;
  \b{\al}_i^*,\b{\al}_{-i}^*,\b{\rho}^*,\mu_n^*)
  ]
  \ \bigg|
 \leq 2C_5'\mu_n^*(n\delta_n'^2
 +\Delta_\al+n\Delta_\rho).
 \end{align}
 By assumption \ref{asymass1} item (5),
 this implies
 $$
 \sup\limits_i
 ||\b{\t{\al}}_i^*-
 \b{\al}_i^*
 ||_2^2 = O(\delta_n'^2+\frac{\Delta_\al}{n}+\Delta_\rho)
 .$$

\subsection{Proof of theorem \ref{asymth6}}
Note that lemma \ref{asymlem6} holds with true parameter
replaced by pseudo true parameter. i.e.,
 for all
$\varepsilon>0$ there exists $C_\varepsilon>0$
such that with probability
at least $1-\varepsilon$
the following holds:
\begin{align}\label{asymeq41}
(\forall i\leq n)
\ \ n\mu_n^*\cdot ||\bhtal_i-\btal^*_i||_2^2
\ \leq  \  C_\varepsilon
\big(
\frac{\sqrt{d}}{\delta_n}||\bhtal-\btal^*||_{2,\mb{Y}_i}^2
 +\mu_n^*\log n\sqrt{n\Delta_\al}+
 \mu_n^* n\log n\sqrt{\Delta_\rho}
+n^{\frac{1}{2}}
  (\mu_n^*)^{\frac{1}{2}}(\log n)^{\frac{3}{2}}
  \big).
\end{align}
Here $\Delta_\al=\sum\limits_{i\leq n}
||\boldsymbol{\h{\t{\al}}}_i-\boldsymbol{\t{\al}}_i^*||_2^2
$,
$ \Delta_\rho= ||\boldsymbol{\h{\t{\rho}}}-\boldsymbol{\t{\rho}}^*||_2^2
$
and
$||\mb{v}||_{2,\mb{Y}_i}^2$ refers
to
$\sum\limits_{j:Y_{ij}\ne \mb{0}}
v_j^2$.
The major
 trouble is
 the term
 $||\bhtal-
 \btal^*||_{1,\mb{Y}_i}$.
By definition of $\delta_n$,
we have,
\begin{align}\nonumber
||\bhtal_i-
 \btal^*_i||_1
 \leq \frac{\sqrt{d}}{\delta_n}
 ||\bhtal_i-
 \btal^*_i||_2^2.
\end{align}
 To simplify notation,
set $\Delta =
\log n\sqrt{\frac{\Delta_\al}{n}}
+\log n \sqrt{\Delta_\rho}+
\frac{(\log n)^{\frac{3}{2}}}
{\sqrt{n\mu_n^*}}$;
set $A$ be the
adjacency matrix of network,
i.e., $ \big(\ I(Y_{ij}\ne\mb{0})
\ \big)_{i,j\leq n}$;
let $u_i =
||\bhtal_i-\btal^*_i||_2^2$,
$\mb{u} = (u_1,\cdots,u_n)$.
Event (\ref{asymeq41}) becomes
\begin{align}\label{asymeq42}
\mb{u}\leq
\frac{\sqrt{d}C_\varepsilon}{n\mu_n^*\delta_n}A\mb{u}+
C_\varepsilon\Delta
.\end{align}
By assumption \ref{asymass1} item (1) for some constant
$C_{\mathcal{B}}$, $u_i\leq C_{\mathcal{B}}$.
Also note that $u_i\ne 0$ implies
$u_i\geq \delta_n$. Let
$\1u =
(I(u_1\ne 0), I(u_2\ne 0), \cdots,I(u_n\ne 0)\ )$.
Therefore (\ref{asymeq42}) becomes
\begin{align}\label{asymeq48}
\1u\leq
\frac{\sqrt{d}C_\varepsilon C_{\mathcal{B}}}
{n\mu_n^*\delta_n^2}A\1u
+\frac{C_\varepsilon\Delta}{\delta_n}.
\end{align}

The point is that if
$\Delta<\frac{\delta_n}
{2C_{\mathcal{B}}}$ and
$u_i\ne 0$, then
there are  at least
$\frac{n\mu_n^*\delta_n^2}
{2\sqrt{d}C_\varepsilon C_{\mathcal{B}}}$
many $j$ in the neighbor of $i$ with
$u_j\ne 0$. But each $j$ will further
forces more nodes $j'$ with non zero
$u_{j'}$.
 Set
$M_> = \{i: u_i\ne 0\}$;
$N_i = \{j:Y_{ij}\ne\mb{0}\}$.
Through this procedure
we are able to show that
$\delta_n|M_>|$,
which is a lower bound of
$||\bhtal-\btal^*||_2^2$, is too large to be
smaller than $\Delta_\al$ if $M_> \ne \emptyset$.
More strictly speaking, if $\Delta<
\frac{\delta_n}
{2C_\varepsilon}$,
then
event
(\ref{asymeq48}) implies that
$(\forall i\in M_>)
|N_i\cap M_>|
\geq \frac{n\mu_n^*\delta_n^2}
{2\sqrt{d}C_\varepsilon C_{\mathcal{B}}}$
.
 Note that
 $\Delta\geq
 \frac{(\log n)^{\frac{3}{2}}}{\sqrt{n\mu_n^*}}$,
 so $\Delta<
\frac{\delta_n}
{2C_\varepsilon}$
  implies that
  the condition of lemma
  \ref{asymlem6},
namely $\conditionofdeltasecond$,
holds.
Therefore by lemma \ref{asymlem6},
 conditional on
 event (\ref{asymeq48}),
 with probability at least
 $1-\frac{1}{(1-\varepsilon)n^2}$:
 $M_>=\emptyset$ or
 $|M_>|>\kupper$.
 But $|M_>|>\kupper$ is impossible since
 $\Delta_\al<\conditionofdeltathird$ .
 In summary, if
   $\Delta_\al<\conditionofdeltathird$,
  $\Delta<\frac{\delta_n}{2C_\varepsilon}$,
  then we have,
  for every $\varepsilon>0$ with
  probability at least $1-\varepsilon-\frac{1}{n^2}$
  $M_>=\emptyset$.
Thus the proof of theorem \ref{asymth6} is
accomplished once lemma \ref{asymlem6}
is proved.
\begin{lemma}\label{asymlem6}
If $\conditionofdeltasecond$,
  then we have,
\begin{align}\nonumber
\mbbP
\bigg(
\ (\exists M\subset\{1,2,\cdots,n\},
1\leq |M|\leq \kupper)
\ (\forall i\in M)
\ |N_i\cap M|\geq
\frac{n\mu_n^*\delta_n^2}{2\sqrt{d}C_\varepsilon C_\mathcal{B}}
\bigg)
\leq \frac{1}{n^2}.
\end{align}

\end{lemma}
\begin{proof}
Note that if
$\kupper<\const1lem6$ then the conclusion
follows trivially. So in
the following proof, assume
$\kupper\geq \const1lem6$,
which implies
$\conditiononmuth6$.

We prove the conclusion
by firstly evaluating
$\mbbP((\forall i\in M)
\ |N_i\cap M|\geq
\frac{n\mu_n^*\delta_n^2}{2\sqrt{d}C_\varepsilon C_\mathcal{B}})$
for a given set
$M\subset\{1,2,\cdots,n\}$
and then apply Borel Cantelli's lemma.
Set
$Z_{ij}=
I(Y_{ij}\ne \mb{0})$,
$Z = \sum\limits_{i<j\in M}
Z_{ij}$.
To evaluate
$\mbbP((\forall i\in M)
\ |N_i\cap M|\geq
\frac{n\mu_n^*\delta_n^2}{2\sqrt{d}C_\varepsilon C_\mathcal{B}})$,
we further divided into two cases:
\begin{enumerate}
\item $|M|\leq
\klowercase2;
$
\item $
\klowercase2 \leq |M|\leq \kupper$.
\end{enumerate}

Firstly we deal with case (1).
We use the following Azuma's inequality:
\begin{lemma}[Azuma's inequality]
\label{asymth2}
Let $Y_t,t\in \{0,1,\cdots,m\}$ be
a martingale with
$|Y_{t+1}-Y_t|\leq B'$ for all
$t\leq m$ and $Y_0=0$. Then we have
\begin{align}
\mbfP(|Y_m|\geq u)
\leq e^{-\frac{u^2}{B'^{2} m}}.
\end{align}
\end{lemma}
Clearly $|N_i\cap M|\geq
\frac{n\mu_n^*\delta_n^2}{2\sqrt{d}C_\varepsilon C_\mathcal{B}}$ for
all $i\in M$ implies
that
$$
\bigg|\ \big\{\
\{i,j\}:
i,j\in M, Y_{ij}\ne\mb{0}
\big\}\bigg|
=\frac{1}{2} \sum\limits_{i\in M}
|N_i\cap M|
\geq |M|
\frac{n\mu_n^*\delta_n^2}{4\sqrt{d}C_\varepsilon C_\mathcal{B}}.
$$
Note that
the partial sum
of $Z_{ij}-\oE[Z_{ij}],i,j\in M$ is a martingale.
Moreover,
since
$|Z_{ij}|\leq 1$, so by assumption
\ref{asymass1} item (2)
$\oE[Z]\leq \frac{|M|(|M|-1)}{2}Cs\mu_n^*$.
Therefore if $|M|
\leq
\klowercase2$,
 then we have
\begin{align}
\mbbP
\bigg(
Z\geq |M|
\frac{n\mu_n^*\delta_n^2}{4\sqrt{d}C_\varepsilon C_\mathcal{B}}
\bigg)\ \leq&
\ \mbbP
\bigg(
Z-\oE[Z]\geq |M|
\frac{n\mu_n^*\delta_n^2}{4\sqrt{d}C_\varepsilon C_\mathcal{B}}
-\frac{|M|(|M|-1)}{2}Cs\mu_n^*\bigg)
\\ \nonumber
&\text{ since in case (1) }
Cs\mu_n^*<1/8\Rightarrow
\frac{|M|-1}{2}Cs\mu_n^*
\leq
\frac{n\mu_n^*\delta_n^2}{8\sqrt{d}C_\varepsilon C_\mathcal{B}}
\\ \nonumber
\leq&
\ \mbbP
\bigg(
Z-\oE[Z]\geq |M|
\frac{n\mu_n^*\delta_n^2}{8\sqrt{d}C_\varepsilon C_\mathcal{B}}
\bigg)
\\ \nonumber
&\text{ by Azuma's inequality and the fact}
\big|Z_{ij}-\oE[Z_{ij}]\ \big|\leq 1
\\ \nonumber
\leq & \exp\big\{
-\dfrac{
 |M|^2
\big(\frac{n\mu_n^*\delta_n^2}{8\sqrt{d}C_\varepsilon C_\mathcal{B}}
\big)^2
}
{
|M|^2
}
\big\}.
\end{align}
Thus applying Borel Cantelli's lemma,
if $k\leq
\klowercase2
$,
\begin{align}\label{asymeq47}
&\mbbP\bigg(\
(\exists M,|M|=k)\
(\forall i\in M)\
|N_i\cap M|\geq
\frac{n\mu_n^*\delta_n^2}{2\sqrt{d}C_\varepsilon C_\mathcal{B}}
\bigg)
\\ \nonumber
\leq&\
n^k \cdot
\exp\big\{
-
 \big(\frac{n\mu_n^*\delta_n^2}{8\sqrt{d}C_\varepsilon C_\mathcal{B}}
\big)^2
\big\}
\\ \nonumber
\leq&\
\exp
\bigg\{
\frac{n\mu_n^*\delta_n^2}{8\sqrt{d}C_\varepsilon C_\mathcal{B}}
\cdot\big[
16\log n-
\frac{n\mu_n^*\delta_n^2}{8\sqrt{d}C_\varepsilon C_\mathcal{B}}
\big]
\bigg\}
\\ \nonumber
&\text{ since by condition of lemma \ref{asymlem6}
 }
\conditionofdelta
\\ \nonumber
\leq& \exp
\{
-4\log n
\}.
\end{align}

Now we deal with case (2).
By Chernoff inequality and independence of $Z_{ij}$,
for every $u,t>0$
\begin{align}\label{asymeq43}
&\mbbP
\big(
Z\geq
u
\big)
\leq
\exp\bigg\{
\sum\limits_{i<j\in M}
\log(\oE[e^{tZ_{ij}}])
-tu
\bigg\}
\\ \nonumber
&\text{\ since by assumption
\ref{asymass1} item (2)  }
\mbbP(Z_{ij} = 1)\leq Cs\mu_n^*
\\ \nonumber
\leq &
\exp\bigg\{
\frac{|M|(|M|-1)}{2}
\log\big[
1+(e^t-1)Cs\mu_n^*
\big]
-tu
\bigg\}.
\end{align}
Set
\begin{align}\nonumber
u \ &= |M|
\frac{n\mu_n^*\delta_n^2}{4\sqrt{d}C_\varepsilon C_\mathcal{B}},
\\ \nonumber
w \ &= \frac{u}{\frac{1}{2}|M|(|M|-1)},
\\ \nonumber
e^t-1 &=
\frac{w/(Cs\mu_n^*) -1}{1-w}.
\end{align}
Continue (\ref{asymeq43}),
\begin{align}\label{asymeq44}
&\frac{|M|(|M|-1)}{2}
\log\big[
1+(e^t-1)Cs\mu_n^*
\big]
-tu\\ \nonumber
 =& \frac{|M|(|M|-1)}{2}
 \bigg[
 \log(\frac{1-Cs\mu_n^*}{1-w})
 -w\log\big(\ \frac{w(\frac{1}{Cs\mu_n^*}-1)}
 {1-w}\ \big)
 \bigg]
 \\ \nonumber
 =&\frac{|M|(|M|-1)}{2}
 \bigg[
 -w\log w-(1-w)\log(1-w)
 +w\log(Cs\mu_n^*)
 +(1-w)\log(1-Cs\mu_n^*)
 \bigg]
 \\ \nonumber
&\text{ since
 in case (2) }
 w\leq \frac{1}{4}
,
 \text{ so }
 -(1-w)\log (1-w)+(1-w)\log(1-Cs\mu_n^*)
 \leq w,
 \\ \nonumber
\leq&\
\frac{|M|(|M|-1)}{2}
\bigg[\
w\big(-\log w +
\log(Cs\mu_n^*)
+1\big)
\ \bigg]
\\ \nonumber
&\ \text{ since }
  |M|\leq \kupper
 \Rightarrow
  e^3
 Cs\mu_n^* \leq w
 \Rightarrow
 -\log w +
\log(Cs\mu_n^*)
+1\leq -2,
\\ \nonumber
\leq&
-2\cdot |M|
\frac{n\mu_n^*\delta_n^2}{4\sqrt{d}C_\varepsilon C_\mathcal{B}}
.\end{align}
Combine (\ref{asymeq43})(\ref{asymeq44}) and
apply Borel Cantelli's lemma,
if $\klowercase2
\leq  k\leq \kupper$
\begin{align}\label{asymeq46}
&\mbbP\bigg(
(\exists M, |M|=k)
\ (\forall i\in M)
|N_i\cap M|\geq
\frac{n\mu_n^*\delta_n^2}{2\sqrt{d}C_\varepsilon C_\mathcal{B}}
\bigg)
\\ \nonumber
\leq&
\exp\big\{
k\log n-
k\cdot
\frac{n\mu_n^*\delta_n^2}{2\sqrt{d}C_\varepsilon C_\mathcal{B}}
\big\}
\\ \nonumber
&\text{ since }
\conditionofdelta
\\ \nonumber
<&
e^{-75k\log n}.
\end{align}

Now combine the conclusion in
case (1) and (2), namely
(\ref{asymeq47})(\ref{asymeq46}), we obviously
have,
$$
\mbbP
\bigg(
(\exists M,
1\leq |M|\leq \kupper)
\ (\forall i\in M)\
|N_i\cap M|\geq
\frac{n\mu_n^*\delta_n^2}{2\sqrt{d}C_\varepsilon C_\mathcal{B}}
\bigg)
\leq
n(e^{-4\log n}+e^{-75\log n})
\leq \frac{1}{n^2}.
$$
The proof is thus accomplished.
\end{proof}

\subsection{Proof of theorem \ref{asymth1}}

We firstly establish some properties
of the density function $g$.
Let $\theta_{ij} = \alpha_i+\beta_j$.
Denote by $g(\cdot;\theta_{ij},\theta_{ji},\rho)$
(or $g(\cdot;\alpha_i,\alpha_j,\beta_i,\beta_j,\rho)$)
the  density function of $Y_{ij}$
in model (\ref{asymmodel1}).
We need the following properties on the
log likelihood of model
(\ref{asymmodel1}).
\begin{proposition}\label{asymprop1}
There exists constant
$\delta,\varepsilon,C_7, C_8,C_9>0$
(depending on $B$) such that
for all $\theta_{1},\theta_{2},\vartheta_{1},
\vartheta_{2}\in [-2B,2B],\rho,\varrho\in [-B,B],
X\in \{0,1\}^2 $ we have:
\begin{enumerate}
\item
$|\theta_{1}-\vartheta_{1}|+|\theta_{2}-\vartheta_{2}|
+|\rho-\varrho|>\delta$ implies
$$
\oE_{Y\sim g(\cdot;\theta_1,\theta_2,\rho)}
(\log g(Y;\theta_1,\theta_{2},\rho))
-\oE_{Y\sim g(\cdot\theta_1,\theta_2,\rho)}
\log g(Y;\vartheta_2,\vartheta_2,\varrho)
>\varepsilon;
$$

\item
$|\theta_{1}-\vartheta_{1}|+|\theta_{2}-\vartheta_{2}|
+|\rho-\varrho|\leq \delta$ implies
\begin{align}
&|\log g(X;\theta_1,\theta_{2},\rho)
-\log g(X;\vartheta_{1},\vartheta_2,\varrho)|
\\ \nonumber
\leq&
C_8 (|\theta_{1}-\vartheta_{1}|+|\theta_{2}-\vartheta_{2}|
+|\rho-\varrho|);
\end{align}
\item
$|\theta_{1}-\vartheta_{1}|+|\theta_{2}-\vartheta_{2}|
+|\rho-\varrho|\leq \delta$ implies
\begin{align}
&\oE_{Y\sim g(\cdot;\theta_1,\theta_2,\rho)}
(\log g(Y;\theta_{1},
\theta_{2},\rho))
-\oE_{Y\sim g(\cdot;\theta_1,\theta_2,\rho)}
(\log g(Y;\vartheta_{1},
\vartheta_{2},\varrho))
\\ \nonumber
\geq& C_7 ((\theta_{1}-\vartheta_{1})^2+
(\theta_{2}-\vartheta_{2})^2
+(\rho-\varrho)^2);
\end{align}

\item $
|\log g(X;\theta_{1},\theta_{2},\rho)|<C_9$.
\end{enumerate}
\end{proposition}
\begin{proof}
Items (1)(2)(4) are trivial.

Fix an arbitrary
$\theta_1,\theta_2,\rho$ and
regard
$\oE_{Y\sim g(\cdot\theta_1,\theta_2,\rho)}
(\log g(Y;\vartheta_1,\vartheta_2,\varrho))$
as a function
in
$\vartheta_1,\vartheta_2,\varrho$.
We prove item (3) by analysing
the Hessian matrix
of
$\oE_{Y\sim g(\cdot\theta_1,\theta_2,\rho)}
(\log g(Y;\vartheta_1,\vartheta_2,\varrho))$,
namely
$H(\vartheta_1,\vartheta_2,\varrho)$, at
$(\vartheta_1,\vartheta_2,
\varrho)=
(\theta_1,\theta_2,\rho)$.
Let
$g_1(X)=
\frac{\partial g(X;
\theta_1,\theta_2,\rho)}
{\partial \theta_1}$,
$g_2(X)=
\frac{\partial g(X;
\theta_1,\theta_2,\rho)}
{\partial \theta_2}$,
$g_3(X)=
\frac{\partial g(X;
\theta_1,\theta_2,\rho)}
{\partial \rho}$
and $g(X) = g(X;\theta_1,\theta_2,\rho)$.
For any $(a_1,a_2)$, we write
$g_i(a_1,a_2),
g(a_1,a_2)$ for $g_i((a_1,a_2)),
g((a_1,a_2))$ respectively.
It is well known that
for any vector
$\mathbf{v}\in\mathbb{R}^3$,
$\mathbf{v}^T
H(\theta_1,\theta_2,\rho) \mathbf{v} =-
\oE_{Y\sim g(\cdot\theta_1,\theta_2,\rho)}
\big[\ (\ v_1
\frac{g_1(Y)}{g(Y)}+
v_2\frac{g_2(Y)}{g(Y)}+
v_3\frac{g_3(Y)}{g(Y)}\ )^2\ \big]$.
Therefore it suffices to prove that
the rank of the following matrix is
$3$:
\begin{align}\label{asymeq18}
\begin{bmatrix}
\frac{g_1(0,0)
}{g(0,0)}, \frac{g_2(0,0)}{g(0,0)},
\frac{g_3(0,0)}{g(0,0)}\\
\frac{g_1(1,0)}{g(1,0)},
\frac{g_2(1,0)}{g(1,0)},
\frac{g_3(1,0)}{g(1,0)}\\
\frac{g_1(0,1)}{g(0,1)},
\frac{g_2(0,1)}{g(0,1)},
\frac{g_3(0,1)}{g(0,1)}\\
\frac{g_1(1,1)}{g(1,1)},
\frac{g_2(1,1)}{g(1,1)},
\frac{g_3(1,1)}{g(1,1)}
\end{bmatrix}.
\end{align}
Let $Z =
1+ e^{\theta_1}+
e^{\theta_{2}}+
e^{\theta_1+\theta_2+\rho}$,
 $x_1= e^{\theta_1}
$,
$x_2= e^{\theta_2}
$ and
$x_3=e^{\rho}
$.
By simple calculation,
\begin{align}\label{asymeq17}
&g_1(a_1,a_2) =
\frac{a_1 Z-x_1(1+x_2x_3)}
{Z^2}
e^{\theta_1 a_1+
\theta_2 a_2+\rho a_1 a_2},\\ \nonumber
&g_2(a_1,a_2) =
\frac{a_2 Z-x_2(1+x_1x_3)}
{Z^2}e^{\theta_1 a_1+
\theta_2 a_2+\rho a_1 a_2},\\ \nonumber
&g_3(a_1,a_2) =
\frac{a_1 a_2
Z-x_1 x_2 x_3}
{Z^2}e^{\theta_1 a_1+
\theta_2 a_2+\rho a_1 a_2}.
\end{align}
Substitute (\ref{asymeq17})
into (\ref{asymeq18})
and set
$
\begin{bmatrix}
\mathbf{g}_1\\
\mathbf{g}_2\\
\mathbf{g}_3
\end{bmatrix}
=\begin{bmatrix}
g_1(0,0), g_2(0,0), g_3(0,0)\\
g_1(1,0), g_2(1,0), g_3(1,0)\\
g_1(0,1), g_2(0,1), g_3(0,1)\\
\end{bmatrix},
$
we have,
\begin{align}
&\begin{bmatrix}
\mathbf{g}_1\\
\mathbf{g}_2\\
\mathbf{g}_3
\end{bmatrix}=
\begin{bmatrix}
-\frac{x_1(1+x_2x_3)}{Z^2},&
-\frac{x_2(1+x_1x_3)}{Z^2},&
-\frac{x_1x_2x_3}{Z^2}\\ \nonumber
\frac{x_1}{Z}
-\frac{x_1^2(1+x_2x_3)}{Z^2},&
-\frac{x_2x_1(1+x_1x_3)}{Z^2},&
-\frac{x_1^2 x_2x_3}{Z^2}\\ \nonumber
-\frac{x_1x_2(1+x_2x_3)}{Z^2},&
\frac{x_2}{Z}-\frac{x_2^2(1+x_1x_3)}{Z^2}
,&-\frac{x_1x_2^2x_3}{Z^2}
\end{bmatrix}.
\end{align}
Note that,
\begin{align}
\begin{bmatrix}
\mathbf{g}_1\\
\mathbf{g}_2-x_1\mathbf{g}_1\\
\mathbf{g}_3-x_2\mathbf{g}_2
\end{bmatrix}
= \begin{bmatrix}
-\frac{x_1(1+x_2x_3)}{Z^2},&
-\frac{x_2(1+x_1x_3)}{Z^2},&
-\frac{x_1x_2x_3}{Z^2}\\
\frac{x_1}{Z},&
0, &0\\ \nonumber
0,&\frac{x_2}{Z},&0
\end{bmatrix}.
\end{align}
Since
$x_1,x_2,x_3,Z\neq 0$ so
$\begin{bmatrix}
\mathbf{g}_1\\
\mathbf{g}_2-x_1\mathbf{g}_1\\
\mathbf{g}_3-x_2\mathbf{g}_2
\end{bmatrix}$ is obviously
a matrix of rank $3$.
Since $g(X)\neq 0$ for all
$X\in \{0,1\}^2$, therefore
the matrix
(\ref{asymeq18}) is of rank $3$.
Thus the conclusion of
item (3) follows.

\end{proof}

Note that
item (1)(2) of assumption \ref{asymass1}
clearly holds for model (\ref{asymmodel1}).
Therefore by theorem \ref{asymth5},
it holds that
$$
\sum\limits_{i<j\leq n}
||g(\cdot;\hat{\theta}_{ij},\hat{\theta}_{ji},\hat{\rho})-
g(\cdot;\theta_{ij}^*,\theta_{ji}^*,\rho^*)||_2^2
=O_p(n^{\frac{3}{2}}(\log n)^2).
$$
 By proposition \ref{asymprop1}
item (3) we have:
\begin{align}\label{asymeq16}
  \sum\limits_{i<j\leq n}
  \min\{(
  \hat{\theta}_{ij}-\theta^*_{ij})^2
  ,\delta\}+
  n^2 \min\{(\rho^*-\hat{\rho})^2,\delta\}
  =O_p(n^{\frac{3}{2}}(\log n)^2).
\end{align}
But $\delta$ is a constant, therefore
(\ref{asymeq16}) is equivalent to
\begin{align}\label{asymeq19}
  \sum\limits_{i<j\leq n}(
  \hat{\theta}_{ij}-\theta^*_{ij})^2+
  n^2 (\rho^*-\hat{\rho})^2
  =O_p(n^{\frac{3}{2}}(\log n)^2).
\end{align}
So it follows directly that
\begin{align}\label{asymeq21}
|\hat{\rho}-\rho^*|
=O_p(n^{-\frac{1}{4}}\log n).
\end{align}

Now
it is easy to deduce the error of
$\hat{\boldsymbol{\alpha}},
\hat{\boldsymbol{\beta}},
\hat{\rho}$.
Let $w^\alpha_i = \hat{\alpha}_i-\alpha^*_{i}
-\Delta \alpha$,
$w^\beta_i = \hat{\beta}_i-\beta^*_{i}
-\Delta \beta$.
For simplicity, we define two distribution
$\mu^\alpha = \frac{1}{n}
\sum\limits_{i\leq n}
\delta_{w^\alpha_i},
\mu^\beta = \frac{1}{n}
\sum\limits_{i\leq n}
\delta_{w^\beta_i}$, i.e.,
$\mu^\alpha,\mu^\beta$ are
uniform distribution on
$\{\hat{\alpha}_i-\alpha^*_{i}-
\Delta\alpha\}_{i\leq n},
\{\hat{\beta}_i-\beta^*_{i}-
\Delta \beta\}_{i\leq n}$, respectively.
Consider the
two random variables $w^\alpha$,
$w^\beta$ subject to distributions
$\mu^\alpha,\mu^\beta$ respectively.

Since $\oE(w^\alpha)=\oE(w^\beta)=0$
and $w^\alpha\perp w^\beta$,
using (\ref{asymeq19}) we have,
\begin{align}
&\oE((w^\alpha)^2)+
\oE((w^\beta)^2)+
(\Delta \alpha+\Delta \beta)^2
=\oE[
\ (w^\alpha+w^\beta+\Delta\alpha+\Delta\beta)^2
\ ]
=O_p(n^{-\frac{1}{2}}(\log n)^2).
\end{align}
This is clearly equivalent to
$(\Delta \alpha+\Delta \beta)^2
=
O_p(n^{-\frac{1}{2}}(\log n)^2)$;
and
$||
\hat{\boldsymbol{\alpha}}-\boldsymbol{\alpha}^*
-\Delta \alpha||_2^2
+||\hat{\boldsymbol{\beta}}-\boldsymbol{\beta}^*
-\Delta \beta||_2^2
=O_p(\sqrt{n}(\log n)^2)
$.

Thus we finished the proof of
all conclusions of theorem \ref{asymth1}.

\subsection{Proof of theorem \ref{asymth3}}

Note that $\boldsymbol{\alpha}^*+\Delta\alpha,
\boldsymbol{\beta}^*-\Delta\alpha$
is a true parameter. It follows from theorem
\ref{asymth1} that
$$
||
\hat{\boldsymbol{\alpha}}-\boldsymbol{\alpha}^*
-\Delta \alpha||_2^2
+||\hat{\boldsymbol{\beta}}-\boldsymbol{\beta}^*
+\Delta \al||_2^2
=O_p(\sqrt{n}(\log n)^2)
$$
By proposition \ref{asymprop1}, $g$
satisfies assumption \ref{asymass1} item (2)(4)(5).
Thus by theorem \ref{asymth4},
\begin{align}
\sup\limits_{i\leq n}\
(\hat{\alpha}_i-\alpha^*_i-\Delta\alpha)^2
+(\hat{\beta}_i-\beta^*_i+\Delta\alpha)^2
=O_p(n^{-\frac{1}{4}}\log n).
\end{align}
But $\hat{\alpha}_1 = \alpha^*_1 = 0$.
So
$\Delta \alpha = O_p(n^{-\frac{1}{8}}(\log n)^{\frac{1}{2}})$.
Thus
\begin{align}
\sup\limits_{i\leq n}\
(\hat{\alpha}_i-\alpha^*_i)^2
+(\hat{\beta}_i-\beta^*_i)^2
=O_p(n^{-\frac{1}{4}}\log n).
\end{align}
The conclusion of theorem \ref{asymth3} follows.

\section{Acknowledgement}
We would like to thank Jing Zhou for comments
on a draft of this work.

\bibliographystyle{apalike}

\bibliography{sbm}

\begin{thebibliography}{}

\bibitem[Amini et~al., 2013]{amini2013pseudo}
Amini, A.~A., Chen, A., Bickel, P.~J., Levina, E., et~al. (2013).
\newblock Pseudo-likelihood methods for community detection in large sparse
  networks.
\newblock {\em The Annals of Statistics}, 41(4):2097--2122.

\bibitem[Balakrishnan et~al., 2011]{balakrishnan2011noise}
Balakrishnan, S., Xu, M., Krishnamurthy, A., and Singh, A. (2011).
\newblock Noise thresholds for spectral clustering.
\newblock In {\em Advances in Neural Information Processing Systems}, pages
  954--962.

\bibitem[Bickel and Chen, 2009]{bickel2009nonparametric}
Bickel, P.~J. and Chen, A. (2009).
\newblock A nonparametric view of network models and newman--girvan and other
  modularities.
\newblock {\em Proceedings of the National Academy of Sciences},
  106(50):21068--21073.

\bibitem[Celisse et~al., 2012]{Celisse2012Consistency}
Celisse, A., Daudin, J.~J., and Pierre, L. (2012).
\newblock Consistency of maximum-likelihood and variational estimators in the
  stochastic block model.
\newblock {\em Electronic Journal of Statistics}, 6(1):1847--1899.

\bibitem[Chatterjee et~al., 2011]{chatterjee2011random}
Chatterjee, S., Diaconis, P., and Sly, A. (2011).
\newblock Random graphs with a given degree sequence.
\newblock {\em The Annals of Applied Probability}, 21(4):1400--1435.

\bibitem[Choi et~al., 2012]{Choi2012Stochastic}
Choi, D.~S., Wolfe, P.~J., and Airoldi, E.~M. (2012).
\newblock Stochastic blockmodels with a growing number of classes.
\newblock {\em Biometrika}, 99(2):273--284.

\bibitem[Condon and Karp, 2001]{Condon2001Algorithms}
Condon, A. and Karp, R.~M. (2001).
\newblock Algorithms for graph partitioning on the planted partition model.
\newblock {\em Random Structures \& Algorithms}, 18(2):116--140.

\bibitem[Dyer and Frieze, 1989]{Dyer1989The}
Dyer, M.~E. and Frieze, A.~M. (1989).
\newblock The solution of some random np-hard problems in polynomial expected
  time.
\newblock {\em Journal of Algorithms}, 10(4):451--489.

\bibitem[Holland and Leinhardt, 1981]{holland1981exponential}
Holland, P.~W. and Leinhardt, S. (1981).
\newblock An exponential family of probability distributions for directed
  graphs.
\newblock {\em Journal of the american Statistical association},
  76(373):33--50.

\bibitem[Jerrum and Sorkin, 1998]{Jerrum1998The}
Jerrum, M. and Sorkin, G.~B. (1998).
\newblock The metropolis algorithm for graph bisection.
\newblock {\em Discrete Applied Mathematics}, 82(1-3):155--175.

\bibitem[Jin, 2015]{Jin2014Fast}
Jin, J. (2015).
\newblock Fast community detection by score.
\newblock {\em Annals of Statistics}, 43(2):57--89.

\bibitem[Jin et~al., 2016]{continuousmembership2016}
Jin, J., Ke, Z.~T., and Luo, S. (2016).
\newblock Estimating network memberships by simplex vertices hunting.

\bibitem[Jog and Loh, 2015]{jog2015information}
Jog, V. and Loh, P.-L. (2015).
\newblock Information-theoretic bounds for exact recovery in weighted
  stochastic block models using the renyi divergence.
\newblock {\em arXiv preprint arXiv:1509.06418}.

\bibitem[Krzakala et~al., 2013]{Krzakala2013Spectral}
Krzakala, F., Moore, C., Mossel, E., Neeman, J., Sly, A., Zdeborov{\'a}, L.,
  and Zhang, P. (2013).
\newblock Spectral redemption in clustering sparse networks.
\newblock {\em Proceedings of the National Academy of Sciences},
  110(52):20935--20940.

\bibitem[Lei et~al., 2015]{Lei2015Consistency}
Lei, J., Rinaldo, A., et~al. (2015).
\newblock Consistency of spectral clustering in stochastic block models.
\newblock {\em The Annals of Statistics}, 43(1):215--237.

\bibitem[Lelarge et~al., 2015]{lelarge2015reconstruction}
Lelarge, M., Massouli{\'e}, L., and Xu, J. (2015).
\newblock Reconstruction in the labelled stochastic block model.
\newblock {\em IEEE Transactions on Network Science and Engineering},
  2(4):152--163.

\bibitem[Newman and Girvan, 2004]{Newman2004Finding}
Newman, M. E.~J. and Girvan, M. (2004).
\newblock Finding and evaluating community structure in networks.
\newblock {\em Physical Review E Statistical Nonlinear \& Soft Matter Physics},
  69(2 Pt 2):026113--026113.

\bibitem[Newman et~al., 2002]{Newman2002Random}
Newman, M. E.~J., Watts, D.~J., and Strogatz, S.~H. (2002).
\newblock Random graph models of social networks.
\newblock {\em Proceedings of the National Academy of Sciences of the United
  States of America}, 99(Suppl 1):2566--2572.

\bibitem[Pritchard et~al., 2000]{Pritchard2000Inference}
Pritchard, J.~K., Stephens, M., and Donnelly, P. (2000).
\newblock Inference of population structure using multilocus genotype data.
\newblock {\em Genetics}, 155(2):574--578.

\bibitem[Robins et~al., 2009]{Robins2009Recent}
Robins, G., Snijders, T., Wang, P., Handcock, M., and Pattison, P. (2009).
\newblock {\em Recent developments in exponential random graph ( p*) models for
  social networks}.
\newblock Springer New York.

\bibitem[Rohe et~al., 2011]{rohe2011spectral}
Rohe, K., Chatterjee, S., and Yu, B. (2011).
\newblock Spectral clustering and the high-dimensional stochastic blockmodel.
\newblock {\em The Annals of Statistics}, 39(4):1878--1915.

\bibitem[Sarkar and Bickel, 2013]{Sarkar2013Role}
Sarkar, P. and Bickel, P.~J. (2013).
\newblock Role of normalization in spectral clustering for stochastic
  blockmodels.
\newblock {\em Annals of Statistics}, 43(3):455--461.

\bibitem[Schlitt and Brazma, 2007]{Schlitt2007Current}
Schlitt, T. and Brazma, A. (2007).
\newblock Current approaches to gene regulatory network modelling.
\newblock {\em BMC Bioinformatics}, 8 Suppl 6(6):S9--S9.

\bibitem[Shi and Malik, 2000]{Shi2000Normalized}
Shi, J. and Malik, J. (2000).
\newblock Normalized cuts and image segmentation.
\newblock {\em IEEE Transactions on pattern analysis and machine intelligence},
  22(8):888--905.

\bibitem[Sonka et~al., 2008]{Sonka2008Image}
Sonka, M., Hlavac, V., and Boyle, R. (2008).
\newblock {\em Image Processing, Analysis, and Machine Vision}.
\newblock Thomson Learning,.

\bibitem[Xu et~al., 2014]{xu2014edge}
Xu, J., Massouli{\'e}, L., and Lelarge, M. (2014).
\newblock Edge label inference in generalized stochastic block models: from
  spectral theory to impossibility results.
\newblock In {\em Proceedings of The 27th Conference on Learning Theory}, pages
  903--920.

\bibitem[Yan et~al., 2016]{yan2016asymptotics}
Yan, T., Leng, C., Zhu, J., et~al. (2016).
\newblock Asymptotics in directed exponential random graph models with an
  increasing bi-degree sequence.
\newblock {\em The Annals of Statistics}, 44(1):31--57.

\bibitem[Yan et~al., 2015]{yan2015asymptotics}
Yan, T., Qin, H., and Wang, H. (2015).
\newblock Asymptotics in undirected random graph models parameterized by the
  strengths of vertices.
\newblock {\em Available at SSRN 2555489}.

\bibitem[Zhao et~al., 2012]{zhao2012consistency}
Zhao, Y., Levina, E., and Zhu, J. (2012).
\newblock Consistency of community detection in networks under degree-corrected
  stochastic block models.
\newblock {\em The Annals of Statistics}, 40(4):2266--2292.

\end{thebibliography}

\end{document}